%% file: article.tex
\documentclass[11pt,reqno,a4paper]{amsart}
\pdfoutput=1
\def\keywords{\smallskip\noindent\textsc{Key Words. }}
\newif\ifomitproofs
\omitproofsfalse
\newif\ifsubmission
\submissionfalse
\newif\ifacm
\acmfalse
\providecommand{\tbl}[2]{\caption{#1}\centering #2}
\usepackage{amsfonts,amstext,amsmath,amssymb,stmaryrd,mathrsfs,calc,amsthm,xspace,mathdots,mathtools,booktabs}

\input{macros}
\usepackage{cmll} %
\usepackage{bussproofs} %
\newcommand{\appref}[1]{\autoref{#1}}
\renewcommand{\vec}[1]{\bar{\mathsf{#1}}}
\newcommand{\IEEEhspace}[1]{}
\usepackage{thmtools,thm-restate}
\newtheorem{theorem}{Theorem}[section]
\newtheorem{lemma}[theorem]{Lemma}
\newtheorem{proposition}[theorem]{Proposition}
\newtheorem{corollary}[theorem]{Corollary}
\newtheorem{fact}[theorem]{Fact}

\theoremstyle{definition}

\theoremstyle{remark}
\newtheorem{claim}{Claim}[theorem]
\newtheorem{remark}[theorem]{Remark}
\usepackage{tikz}
\usetikzlibrary{positioning}
\usetikzlibrary{arrows}
\usetikzlibrary{automata}
\tikzstyle{intermediate}=[draw=gray!90,very thick,fill=gray!20,circle]
\tikzstyle{state}=[intermediate,minimum width=.65cm,inner sep=1pt]
\tikzstyle{rstate}=[state,rectangle,rounded corners=8pt,minimum
      height=.65cm]
\tikzstyle{every node}=[font=\small]
\tikzstyle{every edge}=[draw,->,>=stealth',shorten >=1pt,semithick]
\tikzstyle{accepting}=[accepting by arrow]
\tikzstyle{initial}=[initial by arrow,initial text=]
\usepackage{hyperref}

\providecommand{\urlstyle}[1]{}
\providecommand{\doi}[1]{\href{http://dx.doi.org/#1}{\nolinkurl{doi:#1}}}
\usepackage[numbers]{natbib}
\renewcommand{\cite}{\citep}

\begin{document}
\renewcommand{\sectionautorefname}{Section}
\renewcommand{\subsectionautorefname}{Section}
\renewcommand{\subsubsectionautorefname}[1]{\S}
\title[Non-Elementary Complexities for BVASS, MELL, and
  Extensions]{Non-Elementary Complexities for \mbox{Branching VASS}, MELL, and
  Extensions}
\thanks{Work partially supported by ANR grant ReacHard 11-BS02-001-01.}
\author[R.~Lazi\'c]{Ranko Lazi\'c}
\address{Department of Computer Science, University of Warwick, UK}
\email{lazic@dcs.warwick.ac.uk}
\author[S.~Schmitz]{Sylvain Schmitz}
\address{LSV, ENS Cachan \& CNRS \& INRIA, France}
\email{schmitz@lsv.ens-cachan.fr}
\begin{abstract}
\input{sec-abstract}

\keywords\ Linear logic, vector addition systems, fast-growing
complexity.
\end{abstract}
\maketitle
\section{Introduction}\label{sec-intro}
\input{sec-intro}
\section{Propositional Linear Logic}\label{sec-llw}
\input{sec-llw}

{\renewcommand{\subsection}{\subsubsection}
\subsection{Intuitionistic Linear Logic}
\input{sec-ill}}

\section{Alternating Branching VASS}\label{sec-abvass}
\input{sec-abvass}

\subsection{\abv\ Games}\label{sec-games}
\input{sec-games}

\section{Relationships Between LL and \abv}\label{sec-focus}
\subsection{From LL to \abv}
We proceed in two steps to show a reduction from LL provability to
\abv\ reachability: first, in \autoref{sub-ll-ill}, we recall a
well-known reduction from LL provability to ILZ provability, and
second, in \autoref{sec-ill}, we exhibit a reduction from ILZ
provability to \abv\ reachability.  The outcome will thus be:
\begin{proposition}\label{prop-ll-abvass}
  There are polynomial space reductions:
  \begin{enumerate}
  \item from (affine, resp.\ contractive) LL provability to (lossy,
    resp.\ expansive) \abv\ reachability,
  \item from (affine, resp.\ contractive) MELL provability to (lossy,
    resp.\ expansive) BVASS$_{\vec 0}$ reachability.
  \end{enumerate}
\end{proposition}

\subsubsection{From LL to ILZ}\label{sub-ll-ill}
\input{sec-ll-ill}

\subsubsection{From ILZ to \abv}\label{sec-ill}
\input{sec-ill-abvass}

Our reductions incur an exponential blow-up in the number of
states---however, as we will see with our complexity upper bounds,
this is not an issue, because the main source of complexity in \abv\
is, by far, the dimension of the system, which is here linear in
$|F|$.

\input{sec-abvass-ll}

\section{\textsc{Tower} Upper Bounds}\label{sec-upb}
\input{sec-upb}

\section{\textsc{Tower} Lower Bounds}\label{sec-lowb}

\input{sec-lowb}

\section{\textsc{Ackermann} Upper Bounds}\label{sec-ack}
\label{sec-incr}
\input{sec-incr}

\input{sec-ack}

\section{Concluding Remarks}
\input{sec-concl}

\input{sec-mell}

\bibliographystyle{abbrvnat}
\bibliography{journalsabbr,conferences,llw}

\end{document}

%% file: macros.tex
\DeclarePairedDelimiter{\tup}{\langle}{\rangle}
\DeclarePairedDelimiter{\enc}{\ulcorner}{\urcorner}
\newcommand{\eqby}[1]{\stackrel{\text{{\tiny{#1}}}}{=}}
\newcommand{\eqdef}{\eqby{def}}

\newcommand{\drule}[3]{\dfrac{#1}{#2}{\text{\footnotesize~#3}}}
\newcommand{\rst}[1][]{\xrightarrow{#1\stackrel{?}{=}\vec 0}}
\newcommand{\abv}{ABVASS$_{\vec 0}$}
\newcommand{\maxm}{\mathrm{max}^-}
\newcommand{\maxp}{\mathrm{max}^+}
\newcommand{\Ack}{\ensuremath{\text{\textsc{Ackermann}}}}
\newcommand{\prvs}[1]{\vdash_{\!\!\text{\tiny #1}}}
\newcommand{\jdg}{\mathrel{\triangleright}}

\newcommand{\thetabrule}[1]{\ensuremath{\mathsf{#1}}}
\makeatletter
\providecommand{\Hmake@df@tag@@}[1]{}
\newcommand{\tabrulelabelr}[5][\relax]{
  \Hmake@df@tag@@{#5}%
  \Hy@GlobalStepCount\Hy@linkcounter
  \xdef\@currentHref{equation.\the\c@section.\the\Hy@linkcounter}%
  \Hy@raisedlink{\hyper@anchorstart{\@currentHref}\hyper@anchorend}%
  {\renewcommand{\arraystretch}{1}%
  \frac{ \begin{array}{c} #3 \end{array} }%
       { \begin{array}{c} #4 \end{array} }%
  \@bsphack
  \protected@write\@auxout{}%
    {\string\newlabel{#2}{{{\thetabrule{#5}}}{\thepage}{#1}{\@currentHref}{}}}%
  \:(\thetabrule{#5})}
}
\makeatother

%% file: sec-abstract.tex
We study the complexity of reachability problems on branching
extensions of vector addition systems, which allows us to derive new
non-elementary complexity bounds for fragments and variants of
propositional linear logic.  We show that provability in the
multiplicative exponential fragment is \textsc{Tower}-hard already in
the affine case---and hence non-elementary.  We match this lower bound
for the full propositional affine linear logic, proving its
\textsc{Tower}-completeness.  We also show that provability in
propositional contractive linear logic is \textsc{Ackermann}-complete.

%% file: sec-intro.tex
The use of various classes of counter machines to provide
computational counterparts to propositional substructural logics has
been highly fruitful, allowing to prove for instance:
\begin{itemize}
\item the undecidability of provability in propositional linear logic
  (LL), thanks to a reduction from the halting problem in Minsky
  machines proved by \citet*{lincoln92}, who initiated much of this
  line of work,
\item the decidability of the $\oc$-Horn fragment of multiplicative
  exponential linear logic, proved by \citet{kanovich95} by reduction
  to reachability in vector addition systems,
\item the decidability of provability in {affine} linear logic, first
  shown by \citeauthor{kopylov01} using a notion of vector
  addition games~\citep{kopylov01},
\item the \Ack-completeness of provability in the
  conjunctive implicative fragment of {relevance logic}, proved by
  \citet{urquhart99}, using reductions to and from expansive
  alternating vector addition systems, and
\item the inter-reducibility between provability in
  {multiplicative exponential} linear logic and reachability in a
  model of {branching} vector addition systems, shown by
  \citet*{degroote04}.%
\end{itemize}

\subsection{Alternating Branching VASS}
In this paper, we revisit the correspondences between propositional
linear logic and counter systems with a focus on computational
complexity.  In \autoref{sec-abvass}, we define a model of
\emph{alternating branching vector addition systems} (ABVASS) with
\emph{full zero tests}.  While this model can be seen as an extension
and repackaging of \citeauthor{kopylov01}'s vector games, its
reachability problem enjoys very simple reductions to and from
provability in LL, which are suitable for complexity statements (see
\autoref{sec-focus}).  We prove that:
\begin{itemize}
\item coverability in the top-down, root-to-leaves direction is
  \textsc{Tower}-complete, i.e.\ complete for the class of problems
  that can be solved with time or space resources bounded by a tower
  of exponentials whose height depends elementarily on the input size
  (see \autoref{sec-upb} for the upper bound and \autoref{sec-lowb}
  for the lower bound), and
\item coverability in the bottom-up, leaves-to-root direction and
  so-called ``meet'' and ``zero-jump'' semantics is complete for \Ack,
  i.e.\ complete for resources bounded by the Ackermann function of
  some primitive-recursive function of the
  input\ifomitproofs\else\ (see \appref{sec-ack})\fi.
\end{itemize}

\subsection{Provability in Substructural Logics}
Our complexity bounds for ABVASS translate into the exact same bounds
for provability in fragments and variants of LL:
\subsubsection{Affine Linear Logic\nopunct} 
(LLW) was proved decidable
by~\citet{kopylov01} in~1995 using vector addition games; a
model-theoretic proof was later presented by
\citet{lafont97}.%
\footnote{
Variants of LLW are popular in the literature on implicit complexity;
for instance, \emph{light affine linear logic}~\citep{asperti02} is
known to type exactly the class \textsc{FP} of polynomial time
computable functions.  In this paper we are however interested in the
complexity of the provability problem (for the propositional
fragment), rather than the complexity of normalisation.  In terms of
typed lambda calculi, our results pertain to the complexity of the
type inhabitation problem.}

The best known complexity bounds for LLW are due to
\citet{urquhart98}: by a reduction from coverability in vector
addition systems~\citep{lipton76}, he derives an \textsc{ExpSpace}
lower bound, very far from the \Ack\ upper bound he obtains from
length function theorems for Dickson's
Lemma~\citep[see e.g.][]{FFSS-lics2011}.

\subsubsection{Contractive Linear Logic\nopunct} 
(LLC) was proved
decidable by~\citet{okada99} by model-theoretic methods.%

\Citet{urquhart99} showed the \Ack-completeness of provability in a
fragment of relevance logic, which is also a fragment of intuitionistic
multiplicative additive LLC.  To the best of our knowledge, there are
no known complexity upper bounds for provability in LLC.

\subsubsection{Multiplicative Exponential Linear Logic\nopunct.}
The main open question in this area is whether the multiplicative
exponential fragment (MELL) is decidable.  It is related to many
decision problems, for instance in computational
linguistics~\citep{rambow94,schmitz10}, cryptographic protocol
verification~\citep{verma05}, the verification of parallel
programs~\citep{bouajjani13}, and data
logics~\citep{bojanczyk09,dimino13}%
.

Thanks to the reductions to and from the reachability problem in
branching vector addition systems with states
(BVASS)~\citep{degroote04} and to the bounds of \citet{lazic10}, we
know that provability in MELL is \textsc{2-ExpSpace}-hard.

\subsubsection{Summary of the Complexity Results\nopunct:}
\begin{description}
\item[LLW]\IEEEhspace{3pt} We improve both the lower bound and the
  upper bound of \citet{urquhart98}, and prove that LLW provability is
  complete for \textsc{Tower}.
\item[LLC]\IEEEhspace{3pt} We show that LLC provability is
  \textsc{Ackermann}-complete; the lower bound already holds for the
  multiplicative additive fragment MALLC.
\item[MELL]\IEEEhspace{3pt} Our \textsc{Tower}-hardness result for
  LLW already holds for affine MELL and thus for MELL, which improves
  over the \textsc{2-ExpSpace} lower bound of \citet{lazic10}.
\item[ILL]\IEEEhspace{3pt} All of our complexity bounds also hold for
  provability in the intuitionistic versions of our calculi.  See
  \appref{sec-ill} for details.
\end{description}

%% file: sec-llw.tex
\subsection{Classical Linear Logic}
For convenience, we present here a sequent calculus for classical
propositional linear logic that works with formul\ae\ in negation
normal form and considers one-sided sequents~\citep[see e.g.][]{troelstra92}.

\subsubsection{Syntax}
Propositional linear logic formul\ae\ are defined by the abstract
syntax
\begin{align*}
  A,B ::= &\: a\mid a^\bot\tag{atomic}\\
  &\mid A\parr B\mid A\otimes B\mid\bot\mid\mathbf 1\tag{multiplicative}\\
  &\mid A\with B\mid A\oplus B\mid\top\mid\mathbf 0\tag{additive}\\
  &\mid\oc A\mid\wn A\tag{exponential}
\end{align*}
where $a$ ranges over atomic formul\ae.  We write ``$A^\bot$'' for the
negation normal form of $A$, where negations are pushed to the atoms
using the dualities $A^{\bot\bot}=A$, $(A\parr B)^\bot=A^\bot\otimes
B^\bot$, $\bot^\bot=\mathbf 1$, $(A\with B)^\bot=A^\bot\oplus B^\bot$,
$\top^\bot=\mathbf 0$, and $(\wn A)^\bot=\oc A^\bot$.  We write
``$A\multimap B$'' for the linear implication $A^\bot\parr B$.

\subsubsection{Sequent Calculus}\label{sub-seq}
The rules of the sequent calculus manipulate multisets of formul\ae,
denoted by $\Gamma$, $\Delta$, \dots, so that the exchange rule is
implicit; ``$\wn\Gamma$'' then denotes a multiset of formul\ae\ all
guarded by why-nots: $\wn\Gamma$ is of the form $\wn A_1,\dots,\wn A_n$.
\begin{gather*}
  \drule{}{\vdash A,A^\bot}{init}\qquad
  \drule{\vdash\Gamma,A\quad\vdash\Delta,A^\bot}{\vdash\Gamma,\Delta}{cut}\\[.8em]
  \drule{\vdash\Gamma,A,B}{\vdash\Gamma,A\parr B}{$\parr$}\quad\,
  \drule{\vdash\Gamma,A\quad\vdash \Delta,B}{\vdash\Gamma,\Delta,A\otimes B}{$\otimes$}\quad\,
  \drule{\vdash\Gamma}{\vdash\Gamma,\bot}{$\bot$}\quad\,
  \drule{}{\vdash\mathbf 1}{$\mathbf 1$}\\[.8em]
  \drule{\vdash\Gamma,A\quad\vdash\Gamma,B}{\vdash\Gamma,A\with B}{$\!\with$}\quad
  \drule{\vdash\Gamma,A}{\vdash\Gamma,A\oplus B}{}
  \drule{\vdash\Gamma,B}{\vdash\Gamma,A\oplus B}{$\!\oplus$}\quad
  \drule{}{\vdash\Gamma,\top}{$\!\top$}\\[.8em]
  \drule{\vdash\Gamma,A}{\vdash\Gamma,\wn A}{$\wn$D}\quad\,
  \drule{\vdash\Gamma}{\vdash\Gamma,\wn A}{$\wn$W}\quad\,
  \drule{\vdash\Gamma,\wn A,\wn A}{\vdash\Gamma,\wn A}{$\wn$C}\quad\,
  \drule{\vdash\wn\Gamma,A}{\vdash\wn\Gamma,\oc A}{$\wn$P}
\end{gather*}
The last four rules for exponential formul\ae\ are called
\emph{dereliction}~($\wn$D), \emph{logical weakening}~($\wn$W),
\emph{logical contraction}~($\wn$C), and \emph{promotion}~($\wn$P).

The \emph{cut} rule can be eliminated in this
calculus, which then enjoys the \emph{subformula
  property}: in any rule except cut, the formul\ae\ appearing in the
premises are subformul\ae\ of the formul\ae\ appearing in the
conclusion.

\subsection{Fragments and Variants}
\Citet{lincoln92} established most of the results on the decidability
and complexity of provability in propositional linear logic.  In
particular, the full propositional linear logic (LL) is undecidable,
while its multiplicative additive fragment (MALL, which excludes the
exponential connectives and rules) is decidable in polynomial space.
As mentioned in the introduction, the main open question in this area
is whether the multiplicative exponential fragment (MELL, which
excludes the additive connectives and rules) is decidable.

Regarding related logics, the structural rules of \emph{structural
  weakening}~(W) and \emph{structural contraction}~(C)
\begin{align*}
  \drule{\vdash\Gamma}{\vdash\Gamma,A}{W}\qquad
  \drule{\vdash\Gamma,A,A}{\vdash\Gamma,A}{C}
\end{align*}
give rise to two decidable variants of LL.  If we replace logical
weakening~($\wn$W) by structural weakening~(W), we define
\emph{affine} linear logic (LLW).  If we similarly replace logical
contraction~($\wn$C) by structural contraction~(C), we define
\emph{contractive} linear logic (LLC).  The sequent calculi for LLW
and LLC also enjoy cut elimination and the subformula property for
cut-free proofs.

%% file: sec-ill.tex
Intuitionistic linear logic is essentially obtained from classical
linear logic by restricting its two-sided sequent calculus to
consequents (the right sides of sequents) with at most one formula.
We present here a variant of intuitionistic linear logic with bottom
\citep[\subsectionautorefname~2.5]{troelstra92}, which we will refer
to as ILZ:
\begin{gather*}
  \drule{}{A\vdash A}{init}\qquad
  \drule{\Gamma\vdash A\quad\Delta,A\vdash B}{\Gamma,\Delta\vdash B}{cut}\\[.8em]
  \drule{\Gamma\vdash A\quad\Delta,B\vdash C}{\Gamma,\Delta,A\multimap B\vdash C}{L$_\multimap$}\qquad
  \drule{\Gamma,A\vdash B}{\Gamma\vdash A\multimap B}{R$_\multimap$}\\[.8em]
  \drule{\Gamma,A,B\vdash C}{\Gamma,A\otimes B\vdash C}{L$_\otimes$}\qquad
  \drule{\Gamma\vdash A\quad\Delta\vdash B}{\Gamma,\Delta\vdash A\otimes B}{R$_\otimes$}\\[.8em]
  \drule{}{\bot\vdash}{L$_\bot$}\qquad
  \drule{\Gamma\vdash}{\Gamma\vdash\bot}{R$_\bot$}\\[.8em]
  \drule{\Gamma\vdash A}{\Gamma,\mathbf 1\vdash A}{L$_{\mathbf 1}$}\qquad
  \drule{}{\vdash\mathbf 1}{R$_\mathbf 1$}\\[.8em]
  \drule{\Gamma,A\vdash C\quad\Gamma,B\vdash C}{\Gamma,A\oplus B\vdash C}{L$_\oplus$}\qquad
  \drule{\Gamma\vdash A}{\Gamma\vdash A\oplus B}{~}\drule{\Gamma\vdash B}{\Gamma\vdash A\oplus B}{R$_\oplus$}
\end{gather*}\begin{gather*}%
  \drule{\Gamma,A\vdash C}{\Gamma,A\with B\vdash
    C}~\drule{\Gamma,B\vdash C}{\Gamma,A\with B\vdash C}{L$_\with$}\qquad
  \drule{\Gamma\vdash A\quad\Gamma\vdash B}{\Gamma\vdash A\with
  B}{R$_\with$}\\[.8em]
  \drule{}{\Gamma\vdash\top}{R$_\top$}\\[.8em]
  \drule{\Gamma,A\vdash B}{\Gamma,\oc A\vdash B}{$\oc$D}\quad
  \drule{\Gamma\vdash B}{\Gamma,\oc A\vdash B}{$\oc$W}\quad
  \drule{\Gamma,\oc A,\oc A\vdash B}{\Gamma,\oc A\vdash B}{$\oc$C}\quad
  \drule{\oc\Gamma\vdash A}{\oc\Gamma\vdash\oc A}{$\oc$P}
\end{gather*}
The fragment without $\bot$ is better known as ILL.

\subsubsection{Affine and Contractive Variants}
The intuitionistic versions with bottom ILZW and ILZC and without
bottom ILLW and ILLC of LLW and LLC are respectively obtained by
adding structural weakening and structural contraction:
\begin{align*}
  \drule{\Gamma\vdash B}{\Gamma,A\vdash B}{W}\qquad
  \drule{\Gamma,A,A\vdash B}{\Gamma,A\vdash B}{C}
\end{align*}
As with the sequent calculi for LL, LLW, and LLC, the intuitionistic
calculi for ILZ, ILZW, and ILZC enjoy cut elimination and the
subformula property for cut-free proofs.

\subsubsection{Relevance Logic}
The sequent calculus LR+ considered by \citet{urquhart99} for a
fragment of relevance logic is IMALLC without $\top$, i.e.\ ILZC
restricted to $\{{\multimap},{\otimes},\mathbf 1,{\oplus},{\with}\}$.

%% file: sec-abvass.tex
We define a ``tree'' extension of vector addition systems with states
(VASS) that combines two kinds of branching behaviours: those of
alternating VASS (\autoref{ssub-avass}) and those of branching VASS
(\autoref{ssub-bvass}).  With this combination, we obtain a
reformulation of \citeauthor{kopylov01}'s vector addition
games~\citep{kopylov01}, for which he showed that
\begin{enumerate}
\item the game is inter-reducible with LL
  provability
\item the ``lossy'' version of the game is inter-reducible with LLW
  provability.
\end{enumerate}
We further add \emph{full zero tests} to this model, as they make the
reduction from LL provability straightforward (see
\autoref{sec-focus}) and can easily be removed (see
\autoref{ssub-abvass}).

\subsection{Definitions}\label{sub-abvass}

\subsubsection{Syntax}
An \emph{alternating branching vector addition system with states and
  full zero tests} (\abv) is \ifomitproofs\else syntactically \fi a tuple
$\?A=\tup{Q,d,T_u,T_f,T_s,T_z}$ where $Q$ is a finite set of
\emph{states}, $d$ is a \emph{dimension} in $\+N$, and $T_u\subseteq
Q\times\+Z^d\times Q$, $T_f\subseteq Q^3$, $T_s\subseteq Q^3$ and
$T_z\subseteq Q^2$ are respectively finite sets of \emph{unary},
\emph{fork}, \emph{split} and \emph{full zero test} rules.  We denote
unary rules $(q,\vec u,q_1)$ in $T_u$ with $\vec u$ in $\+Z^d$
by ``$q\xrightarrow{\vec u}q_1$'', fork rules $(q,q_1,q_2)$ in
$T_f$ by ``$q\to q_1\wedge q_2$'', split rules $(q,q_1,q_2)$ in $T_s$
by ``$q\to q_1+q_2$'', and full zero test rules $(q,q_1)$ in $T_z$ by
``$q\rst q_1$''.

\subsubsection{Deduction Semantics}Given an \abv, its semantics is
defined by a deduction system over \emph{configurations} $(q,\vec v)$
in $Q\times\+N^d$:
\begin{gather*}
  \drule{q,\vec v}{q_1,\vec v+\vec u}{unary}
\intertext{where ``$+$'' denotes component-wise addition in $\+N^d$, if
    $q\xrightarrow{\vec u}q_1$ is a rule (and implicitly $\vec
    v+\vec u$ has no negative component, i.e.\ is in $\+N^d$), and}
  \drule{q,\vec v}{q_1,\vec v\quad q_2,\vec v}{fork}\qquad
  \drule{q,\vec v_1+\vec v_2}{q_1,\vec v_1\quad q_2,\vec
    v_2}{split}\qquad \drule{q,\vec 0}{q_1,\vec 0}{full-zero}
\end{gather*}
if $q\to q_1\wedge q_2$, $q\to q_1+q_2$, and $q\rst q_1$ are rules of
the system, respectively, and ``$\vec 0$'' denotes the $d$-vector
$\tup{0,\dots,0}$ with zeroes on every coordinate.  Such a deduction
system can be employed either \emph{top-down} or \emph{bottom-up}
depending on the decision problem at hand (as with tree automata); the
top-down direction will correspond in a natural way to \emph{proof
  search} in propositional linear logic, i.e.\ will correspond to the
consequence to premises direction in the sequent calculus of
\autoref{sub-seq}.

\subsubsection{Example}
\label{ssub-example}
Let $\?A$ be an \abv\ 
with five states ($q_0, q_1, q_2, q_3, q_4$), 
of dimension $3$,
with six unary rules:
\begin{align*}
q_0 & \xrightarrow{\tup{0, 1, 0}} q_1 &
q_1 & \xrightarrow{\tup{0, -1, 2}} q_1 &
q_1 & \xrightarrow{\tup{0, 0, 0}} q_2 \\
q_2 & \xrightarrow{\tup{0, 1, -1}} q_2 &
q_3 & \xrightarrow{\tup{0, 0, 0}} q_0 &
q_3 & \xrightarrow{\tup{-1, -2, 0}} q_4,
\end{align*}
and with one split rule $q_2 \to q_3 + q_3$.
There are no fork rules and no full zero test rules in $\?A$,
and so it is a BVASS (see \autoref{ssub-bvass}).
A depiction of $\?A$ is in \autoref{f:BVASS},
where we write $c, d, d'$ for vector indices $1, 2, 3$ (respectively),
and specify unary rules in terms of increments and decrements.

\begin{figure}
  \ifomitproofs\begin{center}\else\centering\fi
  \begin{tikzpicture}[auto,node distance=1.4cm]
    \node[state](q0){$q_0$};
    \node[state,right=1.5cm of q0](q1){$q_1$};
    \node[state,right=1.2cm of q1,label=right:{$+$}](q2){$q_2$};
    \node[state,right=of q2](q3){$q_3$};
    \node[state,below=.7cm of q3](q4){$q_4$};
    \path[->,every node/.style={font=\footnotesize}]
          (q0) edge node{$\text{{+}{+}} d$} (q1)
          (q1) edge[loop below] node{$\begin{array}{c}
                                      \text{{-}{-}} d;%
                                      \text{{+}{+}} d';%
                                      \text{{+}{+}} d'
                                      \end{array}$} ()
          (q1) edge (q2)
          (q2) edge[loop below] node{$\begin{array}{c}
                                      \text{{-}{-}} d';%
                                      \text{{+}{+}} d
                                      \end{array}$} ()
          (q2) ++(.33,.33) edge[bend left] (q3)
          (q2) ++(.33,-.33) edge[bend right] (q3)
          (q3) edge[bend left=55] (q0)
          (q3) edge node{$\begin{array}{c}
                          \text{{-}{-}} c;%
                          \text{{-}{-}} d;%
                          \text{{-}{-}} d
                          \end{array}$} (q4);
    \draw[semithick] (q2) -- +(.33,.33)
                     (q2) -- +(.33,-.33);
    \draw[color=black!40] (q2) ++(.33,-.33) 
                          arc[start angle=-45,end angle=45,radius=.47];
  \end{tikzpicture}\vspace*{-1.5em}
  \ifomitproofs\end{center}\fi
  \caption{\label{f:BVASS} An example BVASS.}
\end{figure}

From state $q_0$ and with $c, d, d'$ initialised to $5, 0, 0$
(i.e., from a root node labelled by $(q_0, \tup{5, 0, 0})$),
$\?A$ can reach $q_2$ with $d, d'$ having values $2, 0$,
perform the split rule by dividing $c$ and $d$ almost equally
(i.e., branch to two nodes labelled by 
$(q_3, \tup{3, 1, 0})$ and $(q_3, \tup{2, 1, 0})$),
then in both threads reach $q_2$ again with $d, d'$ having values $4, 0$,
perform the split rule as before, and finally reach $q_4$
with $c, d, d'$ having values $1, 0, 0$ in one thread and
$0, 0, 0$ in the remaining three threads.  
See \autoref{fig-ex-run} for the corresponding deduction tree.
\begin{figure}\footnotesize
  \begin{prooftree}\rootAtTop
    \AxiomC{$q_4,\tup{1,0,0}$}
    \UnaryInfC{$q_3,\tup{2,2,0}$}
      \AxiomC{$q_4,\tup{0,0,0}$}
      \UnaryInfC{$q_3,\tup{1,2,0}$}
    \BinaryInfC{$q_2,\tup{3,4,0}$}
    \UnaryInfC{$q_2,\tup{3,3,1}$}
    \UnaryInfC{$q_2,\tup{3,2,2}$}
    \UnaryInfC{$q_2,\tup{3,1,3}$}
    \UnaryInfC{$q_2,\tup{3,0,4}$}
    \UnaryInfC{$q_1,\tup{3,0,4}$}
    \UnaryInfC{$q_1,\tup{3,1,2}$}
    \UnaryInfC{$q_1,\tup{3,2,0}$}
    \UnaryInfC{$q_0,\tup{3,1,0}$}
    \UnaryInfC{$q_3,\tup{3,1,0}$}
        \AxiomC{$q_4,\tup{0,0,0}$}
        \UnaryInfC{$q_3,\tup{1,2,0}$}
          \AxiomC{$q_4,\tup{0,0,0}$}
          \UnaryInfC{$q_3,\tup{1,2,0}$}
        \BinaryInfC{$q_2,\tup{2,4,0}$}
        \UnaryInfC{$q_2,\tup{2,3,1}$}
        \UnaryInfC{$q_2,\tup{2,2,2}$}
        \UnaryInfC{$q_2,\tup{2,1,3}$}
        \UnaryInfC{$q_2,\tup{2,0,4}$}
        \UnaryInfC{$q_1,\tup{2,0,4}$}
        \UnaryInfC{$q_1,\tup{2,1,2}$}
        \UnaryInfC{$q_1,\tup{2,2,0}$}
        \UnaryInfC{$q_0,\tup{2,1,0}$}
        \UnaryInfC{$q_3,\tup{2,1,0}$}
    \BinaryInfC{$q_2,\tup{5,2,0}$}
    \UnaryInfC{$q_2,\tup{5,1,1}$}
    \UnaryInfC{$q_2,\tup{5,0,2}$}
    \UnaryInfC{$q_1,\tup{5,0,2}$}
    \UnaryInfC{$q_1,\tup{5,1,0}$}
    \UnaryInfC{$q_0,\tup{5,0,0}$}
  \end{prooftree}
  \caption{\label{fig-ex-run}A deduction tree in the BVASS of
    \autoref{f:BVASS}.}
\end{figure}

Further reasoning, where we need to consider arbitrarily unequal splits, 
can show that $\?A$ has a deduction tree whose root
is labelled by $(q_0, \tup{m, 0, 0})$ and with the state label at
every leaf being $q_4$ if and only if $m \geq 4$.  In fact, $\?A$ is a
slightly simplified version of the BVASS $\?B_2$
in \autoref{sec-lowb}.%

\subsection{Decision Problems}
Given an \abv\ $\?A$ and a finite set of states $Q_\ell$, we denote by
a \emph{root judgement} ``$\?A,Q_\ell\jdg q,\vec v$'' the fact that
there exists a deduction tree $\?D$ in $\?A$ with root label $(q,\vec
v)$ and leaf labels in $Q_\ell\times\{\vec 0\}$.  We call $\?D$ a
reachability witness for $(q,\vec v)$.  Root judgements can be
derived through the following deduction rules, which will be handy
in proofs:
\begin{gather*}
   \frac{}{\?A,Q_\ell\jdg q_\ell,\vec 0}{\footnotesize\text{ if }q_\ell\in Q_\ell}
   \qquad
   \frac{\?A,Q_\ell\jdg q_1,\vec v+\vec u}
        {\?A,Q_\ell\jdg q,\vec v}{\footnotesize\text{ if }q\xrightarrow{\vec u}q_1}
   \qquad
   \frac{\?A,Q_\ell\jdg q_1,\vec 0}
        {\?A,Q_\ell\jdg q,\vec 0}{\footnotesize\text{ if }q\rst q_1}
   \\[1em]
   \frac{\?A,Q_\ell\jdg q_1,\vec v
         \quad
         \?A,Q_\ell\jdg q_2,\vec v}
    {\?A,Q_\ell\jdg q,\vec v}{\footnotesize\text{ if }q\to q_1\wedge q_2}
   \qquad
   \frac{\?A,Q_\ell\jdg q_1,\vec v_1
         \quad
         \?A,Q_\ell\jdg q_2,\vec v_2}
    {\?A,Q_\ell\jdg q,\vec v_1+\vec v_2}{\footnotesize\text{ if }q\to q_1+q_2}
\end{gather*}

\subsubsection{Reachability} 
Given an \abv\ $\?A$, a finite set of states $Q_\ell$, and a state
$q_r$, the \emph{reachability} problem asks whether $\?A,Q_\ell\jdg
q_r,\vec 0$; we call a reachability witness for $(q_r,\vec 0)$ more
simply a \emph{reachability witness}.

We will see in \autoref{sec-focus} that this reachability problem is
equivalent to provability in LL; the problem is also related to games
played over vectors of natural numbers, see \appref{sec-games}.  It
is however undecidable:
\begin{fact}\label{th-abvass-reach}
  Reachability in \abv\ is undecidable.
\end{fact}
\begin{proof}
  Reachability is already undecidable in the more restricted model of
  AVASS, see \autoref{fc-avass-reach} below.
\end{proof}%

\subsubsection{Lossy Reachability}\label{s:lossy.reach}
In order to obtain decidability, we must weaken the \abv\ model or the
decision problem.  For the former, let us denote by $\vec e_i$ the
unit vector in $\+N^d$ with one on coordinate $i$ and zero everywhere else.
Then a \emph{lossy} \abv\ can be understood as featuring a rule
$q\xrightarrow{-\vec e_i}q$ for every $q$ in $Q$ and $0<i\leq d$.  We
rather define it by extending its deduction system with
\begin{gather*}
  \drule{q,\vec v}{q,\vec v-\vec e_i}{loss}
\end{gather*}
for every $q$ in $Q$ and $0<i\leq d$.  We write `$\jdg_\ell$' for root
judgements where losses can occur.
In terms of proof search in linear logic, losses will correspond to
structural weakening, which is the distinguishing feature of affine
linear logic.

\paragraph*{Top-Down Coverability}
An alternative way to see the reachability problem in lossy \abv\ is
to weaken the problem.  Let us define a variant of \abv\ that feature
\emph{full resets} instead of full zero tests: we denote in this case
rules $(q,q_1)$ in $T_z$ by $q\xrightarrow{:=\vec0}q_1$ and associate
a different semantics:
\begin{equation*}
  \drule{q,\vec v}{q_1,\vec 0}{full-reset}
\end{equation*}
We call the resulting model ABVASSr.  Given an ABVASSr~$\?A$, a state
$q_r$, and a finite set of states $Q_\ell$, the \emph{top-down
  coverability} or \emph{leaf coverability} problem asks whether there
exists a deduction tree $\?D$ with root label $(q_r,\vec 0)$ and such
that, for each leaf, there exists some $q_\ell$ in $Q_\ell$ and some
$\vec v$ in $\+N^d$ such that the leaf label is $(q_\ell,\vec v)$; we
then call $\?D$ a \emph{coverability witness}.

The reachability problem for lossy \abv\ is then equivalent to
top-down coverability for ABVASSr.  Observe indeed that the unary,
fork, and split rules are \emph{monotone}: if $\vec v\leq\vec w$ for
the product ordering, i.e.\ if $\vec v(i)\leq\vec w(i)$ for all
$0<i\leq d$, and a configuration $(q,\vec v)$ allows to apply a rule
and result in some configurations $(q_1,\vec v_1)$ and (possibly)
$(q_2,\vec v_2)$, then $(q,\vec w)$ allows to apply the same rule and
to obtain some $(q_1,\vec w_1)$ and $(q_2,\vec w_2)$ with $\vec
v_1\leq\vec w_1$ and $\vec v_2\leq\vec w_2$.  This means that losses
in an \abv\ can be applied as late as possible, either right before a
full zero test or at the leaves---which corresponds exactly to
top-down coverability for ABVASSr.

\subsubsection{Expansive Reachability}\label{ssub-rcov}
In order to model structural contractions during proof search, it is
natural to consider another variant of \abv\ called \emph{expansive}
\abv\ and equipped with the deduction rules
\begin{gather*}
  \drule{q,\vec v+\vec e_i}{q,\vec v+2\vec e_i}{expansion}
\end{gather*}
for every $q$ in $Q$ and $0<i\leq d$.  We write `$\jdg_e$' for root
judgements where expansions can occur.
This is a restriction over \abv\ since expansions can be
emulated through two unary rules $q\xrightarrow{-\vec
  e_i}q'\xrightarrow{2\vec e_i}q$.  Expansive reachability is not
quite dual to lossy reachability---we deal with \emph{increasing
  reachability} in \appref{sec-ack}.

\subsection{Restrictions}

Note that \abv\ generalise \emph{vector addition systems with states}
(VASS), which are \abv\ with only unary rules.  They also generalise
two ``branching'' extensions of VASS, which have been defined in
relation with propositional linear logic.  Since these restrictions do
not feature full zero tests, their lossy reachability problem is
equivalent to their top-down coverability problem.

\makeatletter\let\@period\@empty\makeatother
\subsubsection{Alternating VASS}\label{ssub-avass}
were originally
called ``and-branching'' counter machines by \citet{lincoln92}, and
were introduced to prove the undecidability of propositional linear
logic.  Formally, an AVASS is an \abv\ which only features unary and
fork rules, i.e.\ with $T_s=T_z=\emptyset$.
\begin{fact}[\citet{lincoln92}]\label{fc-avass-reach}
  Reachability in AVASS is undecidable.
\end{fact}
\begin{proof}[Proof Idea]
  By a reduction from the halting problem in Minsky machines: note
  that a zero test $q\xrightarrow{c\stackrel{?}{=}0}q'$ on a counter
  $c$ can be emulated through a fork $q\to q'\wedge q_c$, where unary
  rules $q_c\xrightarrow{-\vec{e}_{c'}}q_c$ for all $c'\neq c$ allow
  to empty the counters different from $c$, and a last unary rule
  $q_c\xrightarrow{\vec{0}} q_\ell$ to the single target state allows
  to check that $c$ was indeed equal to zero.
\end{proof}
Alternating VASS do not allow to model LL proof search in full;
\citet{kanovich95} identified the matching LL fragment, called the
\mbox{$({!},{\oplus})$-Horn} fragment.

The complexity of the other basic reachability problems on AVASS is
known:
\begin{itemize}
\item motivated by the complexity of fragments of relevance logic,
  \citet{urquhart99} proved that expansive reachability is complete
  for Ackermannian time, and
\item motivated by the complexity of vector addition games (see
  \appref{sec-games}), \citet{courtois14} showed that lossy
  reachability is \textsc{2-ExpTime}-complete.
\end{itemize}

\subsubsection{Branching VASS\nopunct.}\label{ssub-bvass}  Inspired by the
correspondences between the \mbox{$\oc$-Horn} fragment of linear logic and
VASS unearthed by \citet{kanovich95}, \citet{degroote04} defined
BVASS---which they originally dubbed ``vector addition tree
automata''---as a model of counter machines that matches MELL.
Formally, a BVASS is an \abv\ with only unary and split rules, i.e.\
with $T_f=T_z=\emptyset$.  This model turned out to be equivalent to
independently defined models in linguistics~\citep{rambow94} and
protocol verification~\citep{verma05}; see \citep{schmitz10} for a
survey.

Whether BVASS reachability is decidable is an open problem, and is
inter-reducible with MELL provability.  \Citet{lazic10} proved the best
known lower bound to this day, which is \textsc{2-ExpSpace}-hardness.
Two related problems were shown to be \textsc{2-ExpTime}-complete
by \citet{demri12}, namely increasing reachability (see
\appref{sec-ack}) and boundedness.

\subsubsection{Alternating Branching VASS\nopunct.}\label{ssub-abvass}
\Citet{kopylov01} defined a one-player vector game, which matches
essentially the reachability problem in ABVASS, i.e.\ in \abv\ with
$T_z=\emptyset$.  The \emph{elementary} fragment of ILL defined by
\citet{larchey13} is another counterpart to ABVASS.

While allowing full zero tests is helpful in the reduction from
LL provability, they can be dispensed with at little expense.
Let us first introduce some notation.  If node $n$ is an ancestor of a
node $n'$ in a deduction tree $\?D$, and the labels of $n$ and $n'$
are the same, we write $\?D[n \leftarrow n']$ for the
\emph{shortening} of $\?D$ obtained by replacing the subtree of rule
applications rooted at $n$ by the one rooted at $n'$.  Observe that,
if $\?D$ is a reachability witness (resp.\ a coverability witness),
then $\?D[n\leftarrow n']$ is also a reachability witness (resp.\ a
coverability witness).
\begin{lemma}\label{lem-abvass}
  There is a logarithmic-space reduction from (lossy, resp.\
  expansive) \abv\ reachability to (lossy, resp.\ expansive) ABVASS
  reachability.
\end{lemma}
\begin{proof}
  Suppose $\?A$ is an \abv\ with set of states $Q$ and dimension $d$.

  For a logarithmic-space many-one reduction, the key observation is
  that, if there exists a witness for an instance of (lossy, resp.\
  expansive) reachability for $\?A$, then by repeated shortenings,
  there must be one in which, along every vertical path, the number of
  occurrences of full zero tests is at most $|Q| - 1$.

  It therefore suffices to decide the problem for an ABVASS $\?A^\dag$ 
  whose set of states is $\{1, \ldots, |Q|\} \times Q$, 
  whose dimension is $|Q| \cdot d$, and which simulates $\?A$
  up to $|Q| - 1$ full zero tests along any vertical path.
  In any state $(i, q)$, $\?A^\dag$ behaves like $\?A$ in state $q$,
  but using the $i$th $d$-tuple of its vector components.
  To simulate a full zero test $q\rst q'$ in $\?A$, 
  $\?A^\dag$ changes state from $(i, q)$ to $(i + 1, q)$,
  postponing the check that the $i$th $d$-tuple of vector components
  are zero until the leaves \ifomitproofs\else of the deduction tree\fi.
\end{proof}
\begin{remark}[Polynomial Time Turing Reduction]
  There is a polynomial time Turing reduction to the same effect.  Its
  interest is that it preserves the dimension of the \abv.  Because
  the dimension is---by far---the most important source of complexity
  in our upper bounds, preserving it might be useful in some
  circumstances.

  Let us first define the
  set of \emph{root states} relative to a subset $X$ of $Q$
  by \begin{equation} \mathrm{Root}_{\?A}(X)\eqdef\{q\in
  Q\mid\?A,X\jdg q,\vec 0\} \end{equation} as the set of states $q$
  such that there exists a deduction in $\?A$ with root label $(q,\vec
  0)$ and leaf labels in $X\times\{\vec 0\}$.  The (lossy, resp.\
  expansive) reachability problem for $\tup{\?A,q_r,Q_\ell}$ then
  reduces to checking whether $q_r$ belongs to
  $\mathrm{Root}_\?A(Q_\ell)$.

  \ifomitproofs
\else Let $\?A=\tup{Q,d,T_u,T_f,T_s,T_z}$.  Writing
  $\?A'$ for the ABVASS $\tup{Q,d,T_u,T_f,T_s,\emptyset}$, \fi
  we can compute $\mathrm{Root}_{\?A'}(X)$ using $|Q|$ calls to an
  oracle for (lossy, resp.\ expansive) ABVASS reachability.  Moreover,
  since $\mathrm{Root}_{\?A'}(X)\supseteq X$ is monotone, we can use a
  least fixed point computation that discovers root states according
  to the number of full zero tests along the branches of their
  reachability witnesses:
  \begin{equation}
    \mathrm{Root}_\?A(Q_\ell)=\mu
    X.\mathrm{Root}_{\?A'}(Q_\ell)\cup\mathrm{Root}_{\?A'}(X\cup T_z^{-1}(X))\,.
  \end{equation}
  This computation converges after at most $|Q|$ steps, and therefore
  works in polynomial time relative to the same oracle.
\end{remark}

\subsection{Computational Complexity}

\subsubsection{Non-Elementary Complexity Classes\nopunct.}
We will use in this paper two complexity classes \citep[see][]{arXiv/Schmitz13}:
\begin{align}
  \text{\textsc{Tower}}&\eqdef{\bigcup_{e\in\text{\textsc{FElem}}}}\text{\textsc{DTime}}\big(\mathrm{tower}(e(n))\big)
    \intertext{is the class of problems that can be solved with a
      deterministic Turing machine in time $\mathrm{tower}$ of some
      elementary function $e$ of the input, where
      $\mathrm{tower}(0)\eqdef 1$
      and $\mathrm{tower}(n+1)\eqdef2^{\mathrm{tower}(n)}$ defines
      towers of exponentials.
      Similarly,}
  \Ack&\eqdef{\bigcup_{p\in\text{\textsc{FPR}}}}\text{\textsc{DTime}}\big(\mathrm{Ack}(p(n))\big)
\end{align}
is the class of problems solvable in time $\mathrm{Ack}$ of some
primitive recursive function $p$ of the input size, where
``$\mathrm{Ack}$'' denotes the Ackermann function---any standard
definition of $\mathrm{Ack}$ yields the same complexity
class~\citep{arXiv/Schmitz13}.

Completeness for $\textsc{Tower}$ is understood relative to many-one
elementary reductions, and completeness for \Ack\ relative to many-one
primitive-recursive reductions.%

\subsubsection{\abv\ Complexity\nopunct.}
For a set $T_u$ of unary rules, we write $\maxm(T_u)$
(resp.\ $\maxp(T_u)$) for the largest absolute value of any negative
(resp.\ positive) integer in a vector in $T_u$, and
$\mathrm{max}(T_u)$ for their overall maximum.  We assume a binary
encoding of the vectors in unary rules, thus $\mathrm{max}(T_u)$ might
be exponential in the size of the \abv.  We can however reduce to
\emph{ordinary} \abv, i.e.\ \abv\ with $\vec u=\vec e_i$ or $\vec
u=-\vec e_i$ for some $0<i\leq d$ whenever $q\xrightarrow{\vec u}q_1$
is a unary rule:
\begin{lemma}\label{lem-ordinary}
  There is a logarithmic space reduction from reachability in (lossy,
  resp.\ expansive) \abv\ to reachability in (lossy, resp.\ expansive)
  ordinary \abv.
\end{lemma}
\begin{proof}[Proof Idea]
  The idea is to encode each of the $d$ coordinates of the original
  \abv\ into $\lfloor\log(\mathrm{max}(T_u)+1)\rfloor$ coordinates, and
  each unary rule to apply a binary encoding of $\vec u$ to those new
  coordinates; see for instance \citep{schmitz10} where this
  construction is detailed for BVASS.  The expansive case requires to
  first explicitly encode expansions as unary rules.
\end{proof}

\paragraph{Lossy Case\nopunct.} One of the main results of this paper is the following:
\begin{theorem}\label{th-bvass}
  Reachability in lossy BVASS and lossy \abv\ is
  \textsc{Tower}-complete.
\end{theorem}
\begin{proof}
  The upper bound is proved in \autoref{sec-upb}.  We present the
  hardness proof in detail in \autoref{sec-lowb}.
\end{proof}\noindent
Note that \autoref{th-bvass} entails an improvement for BVASS
reachability over the \textsc{2-ExpSpace} lower bound of
\citet{lazic10}.%

\paragraph{Expansive Case\nopunct.}
Regarding expansive \abv, we can adapt the proofs of \citet{urquhart99}
for expansive AVASS and the relevance calculus LR+ to show:
\begin{restatable}{theorem}{abvasscov}\label{th-abvass-cov}
  Reachability in expansive AVASS and expansive \abv\ is
  \Ack-complete.
\end{restatable}
\begin{proof}
  The lower bound is due to \citet{urquhart99}, who proved hardness of
  expansive AVASS reachability by a direct reduction from the halting
  problem of Minsky machines with counter values bounded by the
  Ackermann function.  The upper bound can be proved following
  essentially the same arguments as \citeauthor{urquhart99}'s for
  LR+, using length function theorems for Dickson's Lemma~\citep[see
  e.g.][]{FFSS-lics2011}.  See \appref{sec-ack} for a proof.
\end{proof}
\autoref{th-abvass-cov} allows to derive the same \textsc{Ackermann}
bounds for provability in MALLC and LLC, see \appref{sec-ack}.

%% file: sec-games.tex
Reachability problems in \abv\ can also be understood in a
game-theoretic setting.  Let us fix a reachability instance
$\tup{\?A,q_r,Q_\ell}$ and consider the following zero-sum two players
game over the infinite arena $Q\times\+N^d$, where $d$ is the
dimension of $\?A$: its two players are called \emph{Controller} and
\emph{Environment}.  The game starts in the configuration $(q_r,\vec
0)$.  In a current configuration $(q,\vec v)$, Controller chooses a
rule of $\?A$, which allows to apply one of the deduction rules of
$\?A$, or loses if none applies.  In the case of a unary or full zero
test rule, the current configuration changes to $(q_1,\vec v+\vec u)$
and $(q_1,\vec 0)$ respectively.  In the case of a fork rule,
Environment chooses between a move to $(q_1,\vec v)$ or $(q_2,\vec
v)$.  In the case of a split rule, Controller furthermore chooses two
vectors $\vec v_1,\vec v_2$ in $\+N^d$ with $\vec v_1+\vec v_2=\vec v$
and Environment chooses between a move to $(q_1,\vec v_1)$ or a move to
$(q_2,\vec v_2)$.

The objective of Controller is to reach a configuration $(q_\ell,\vec
0)$ with $q_\ell$ in $Q_\ell$; the objective of Environment is to
prevent it.  It is easy to see that Controller has a winning strategy
if and only if the original reachability instance was positive.

Increasing, expanding, or lossy reachability are straightforward to
handle in this game setting.  Interestingly, in the case of lossy
reachability, we can take the full-reset semantics for $T_z$, and
Controller's objective can then be restated as reaching $(q_\ell,\vec
v)$ for some $q_\ell$ in $Q_\ell$ and some vector $\vec v$ in $\+N^d$,
i.e.\ as a state reachability objective.  This game view is related to
\emph{multi-dimensional energy
  games}~\citep{brazdil10,chatterjee10,abdulla13}, which are played on
AVASS (defined in \autoref{ssub-avass}).

%% file: sec-ll-ill.tex
\label{sec-ll-ill}
The Kolmogorov translation of classical logic into intuitionistic
logic by double negation can be adapted to linear logic~\citep{troelstra92}:
\begin{fact}\label{fc-ll-ill}
  There is a polynomial time reduction from (affine, resp.\
  contractive) LL provability to (affine, resp.\ contractive) ILZ
  provability.
\end{fact}
\begin{proof}[Proof Idea]
  A translation of classical linear formul\ae\ $A$ into intuitionistic
  ones $A^k$ is provided by
  \citet[\subsectionautorefname~5.12]{troelstra92}, which satisfies
  \begin{equation}\label{eq-ll-ill}
    \Gamma\prvs{LL}A\quad\text{iff}\quad\Gamma^k\prvs{ILZ}A^k\;.
  \end{equation}
  The proof of this fact uses in particular that, for all $A^k$,
  \begin{equation}\label{eq-quadneg}
    \prvs{ILZ}((A^k\multimap\bot)\multimap\bot)\multimapboth A^k\;.
  \end{equation}

  We merely need to check that the result also holds in presence of
  structural weakenings or structural contractions.  Because the
  sequent calculi for ILZW and ILZC are restrictions of those for LLW
  and LLC, we only need to exhibit a translation of the structural
  rules of the two-sided sequent calculus for LLW and LLC into ILZW
  and ILZC proofs.

  For structural weakenings, the two-sided sequent calculus for LLW
  has rules
  \begin{equation*}
    \drule{\Gamma\vdash\Delta}{\Gamma,A\vdash\Delta}{L$_W$}
    \qquad
    \drule{\Gamma\vdash\Delta}{\Gamma\vdash A,\Delta}{R$_W$}
  \end{equation*}
  In order to prove the affine version of \eqref{eq-ll-ill}, for
  (L$_W$) we need to restrict ourselves to $\Delta=B$ a single
  formula, and see that this is exactly the structural weakening of
  ILZW.  For (R$_W$) we need to restrict ourselves to an empty
  $\Delta$: then $\Gamma\prvs{LLW}$
  if and only if $\Gamma\prvs{LLW}\bot$ and\vspace*{-.5em}
  \begin{prooftree}
    \AxiomC{$\Gamma^k\prvs{ILZW}\bot$}
    \LeftLabel{{\tiny (L$_W$)}}
    \UnaryInfC{$\Gamma^k,A^k\multimap\bot\prvs{ILZW}\bot$}
    \LeftLabel{{\tiny (R$_\multimap$)}}
    \UnaryInfC{$\Gamma^k\prvs{ILZW} (A^k\multimap\bot)\multimap\bot$}
  \end{prooftree}
  which, together with 
  $\prvs{ILZW}((A^k\multimap\bot)\multimap\bot)\multimapboth
  A^k$ for all $A^k$ by \eqref{eq-quadneg}, allows to conclude.

  For structural contractions, the two-sided sequent calculus for LLC
  has rules
  \begin{equation*}
    \drule{\Gamma,A,A\vdash\Delta}{\Gamma,A\vdash\Delta}{L$_C$}
    \qquad
    \drule{\Gamma\vdash A,A,\Delta}{\Gamma\vdash A,\Delta}{R$_C$}
  \end{equation*}
  Again, for (L$_C$) we have $\Delta=B$ a single formula and it turns
  out to be exactly the structural contraction of ILZC.  For (R$_C$),
  necessarily $\Delta$ is empty.  Then $\Gamma\prvs{LLC}A,A$ if and
  only if $\Gamma,A\multimap\bot\prvs{LLC}A$ and\vspace*{-.5em}
  \begin{prooftree}
    \AxiomC{$\Gamma^k,A^k\multimap\bot\prvs{ILZC}A^k$}
    \AxiomC{}
    \RightLabel{{\tiny (L$_\bot$)}}
    \UnaryInfC{$\bot\prvs{ILZC}$}
    \LeftLabel{{\tiny (L$_\multimap$)}}
    \BinaryInfC{$\Gamma^k,A^k\multimap\bot,A^k\multimap\bot\prvs{ILZC}$}
    \LeftLabel{{\tiny (R$_\bot$)}}
    \UnaryInfC{$\Gamma^k,A^k\multimap\bot,A^k\multimap\bot\prvs{ILZC}\bot$}
    \LeftLabel{{\tiny (L$_C$)}}
    \UnaryInfC{$\Gamma^k,A^k\multimap\bot\prvs{ILZC}\bot$}
    \LeftLabel{{\tiny (R$_\multimap$)}}
    \UnaryInfC{$\Gamma^k\prvs{ILZC}(A^k\multimap\bot)\multimap\bot$}
  \end{prooftree}
  which, together with $\prvs{ILZC}((A^k\multimap\bot)\multimap\bot)\multimapboth A^k$
  for all $A^k$ by \eqref{eq-quadneg}, allows to conclude.
\end{proof}

%% file: sec-ill-abvass.tex
The key property of the sequent calculi for ILZ, ILZW, and ILZC we
exploit in our reduction to \abv\ is the subformula property of
cut-free proofs.

Let us consider an instance of the provability problem for ILZ,
i.e.\ some formula $F$.  The subformula property allows us to view a
sequent $\oc\Psi,\Delta\vdash A$ appearing in a cut-free proof of
$\vdash F$ as a triple consisting of a multiset $\oc\Psi$ of
$\oc$-guarded subformul\ae, a multiset $\Delta$ of
subformul\ae, and a subformula $A$---all subformul\ae\ of the target
formula $F$.  Thanks to logical weakening ($\oc$W) and logical
contraction ($\oc$C), we will even be able to treat $\oc\Psi$ as a
\emph{set}.

Let us write $S$ for the set of subformul\ae\ of $F$ and
$S_\oc\subseteq S$ for its $\oc$-guarded subformul\ae.  We define from
$F$ an \abv\ $\?A^I_F$ of dimension $|S|$ that includes
$2^{S_\oc}\times (S\uplus\{.\})$ in its state space, where ``$.$'' is
a fresh symbol.  A configuration of $\?A^I_F$ in $2^{S_\oc}\times
S\times\+N^S$ encodes a sequent $\oc\Psi,\Delta\vdash A$ as
$(\sigma(\oc\Psi),A,\Delta)$, where $\sigma$ associates to a multiset
its \emph{support}, i.e.\ its set without duplicates---note that we
completely identify multisets in $\+N^S$ with vectors in $\+N^{|S|}$.
A configuration in $(2^{S_\oc}\times\{.\}\times\+N^S)$ encodes a
sequent $\oc\Psi,\Delta\vdash$ as $(\sigma(\oc\Psi),.,\Delta)$.  We
also include a distinguished leaf state $q_\ell$ in the state space of
$\?A^I_F$.  The rules of $\?A^I_F$ implement the rules of ILZ on the
encoded configurations in a straightforward manner---they rely on
additional intermediate states for this---and are depicted in
\autoref{fig-ill-abvass}.  An additional \emph{store} rule allows to
move an ``of-course'' $\oc A$ formula from counters to state storage;
we could compile it into the other rules at the expense of a larger
number of cases.

\begin{figure}
  \centering
  \begin{tikzpicture}[auto,node distance=.65cm]
    \node[state,label=left:{init:},label=below:{{\footnotesize $A\not\in S_\oc$}}](init0){$\emptyset,A$};
    \node[state,accepting by double,right=of init0](init1){$q_\ell$};
    \node[rstate,right=of init1](init2){$\{\oc A\},\oc A$};
    \node[state,accepting by double,right=of init2](init3){$q_\ell$};    
    \path (init0) edge node{$-\vec e_{A}$} (init1)
    (init2) edge node{$\vec 0$} (init3);
    \node[state,right=1.47cm of
      init3,label=left:{store:},label=below:{{\footnotesize $B\in S\uplus\{.\}$}}](store0){$q,B$};
    \node[rstate,right=.8cm of store0](store1){$q\cup\{\oc A\},B$};
    \path (store0) edge node{$-\vec e_{\oc A}$} (store1);
    \node[rstate,below right=1.3cm and -.7cm of init0,label=left:{L$_\multimap$:}](Limp0){$q\cup q',C$};
    \node[intermediate,right=1.1cm of Limp0,label=right:{$+$}](Limp1){};
    \draw[draw=black!40] (Limp1) -- ++(.2,-.2) arc (-40:40:.33);
    \node[state,above right=of Limp1](Limp2){$q,A$};
    \node[intermediate,below right=of Limp1](Limp3){};
    \node[state,right=of Limp3](Limp4){$q',C$};
    \path (Limp0) edge node{$-\vec e_{A\multimap B}$} (Limp1)
      (Limp1) edge (Limp2)
      (Limp1) edge (Limp3)
      (Limp3) edge node{$\vec e_B$}(Limp4);
    \node[rstate,right=4.2cm of
      Limp1,label=left:{R$_\multimap$:}](Rimp0){$q,A\multimap B$};
    \node[state,right=of Rimp0](Rimp1){$q,B$};
    \path (Rimp0) edge node{$\vec e_A$} (Rimp1);
    \node[state,below=3.2cm of
      init0,label=left:{L$_\otimes$:}](Lotimes0){$q,C$};
    \node[intermediate,right=1cm of Lotimes0](Lotimes1){};
    \node[state,right=1.2cm of Lotimes1](Lotimes2){$q,C$};
    \path (Lotimes0) edge node{$-\vec e_{A\otimes B}$} (Lotimes1)
      (Lotimes1) edge node{$\vec e_A+\vec e_B$} (Lotimes2);
    \node[rstate,right=2.85cm of Lotimes2,
      label=left:{R$_\otimes$:},
      label=right:{$\,+$}](Rotimes0){$q\cup q',A\otimes B$};
    \draw[draw=black!40] (Rotimes0.east) -- ++(.1,-.1) arc (-40:40:.16);
    \node[state,above right=.2cm and .3cm of Rotimes0](Rotimes1){$q,A$};
    \node[state,below right=.2cm and .3cm of Rotimes0](Rotimes2){$q',B$};
    \path (Rotimes0.east) edge (Rotimes1)
      (Rotimes0.east) edge (Rotimes2);
    \node[state,below=1cm of
      Lotimes0,label=left:{L$_\bot$:}](Lbot0){$q,.$};
    \node[state,accepting by double,right=of Lbot0](Lbot1){$q_\ell$};
    \path (Lbot0) edge node{$-\vec e_\bot$} (Lbot1);
    \node[state,right=4.96cm of Lbot1,label=left:{R$_\bot$:}](Rbot0){$q,\bot$};
    \node[state,right=of Rbot0](Rbot1){$q,.$};
    \path (Rbot0) edge node{$\vec 0$} (Rbot1);
    \node[state,below=of Lbot0,label=left:{L$_{\mathbf
        1}$:}](Lone0){$q,A$};
    \node[state,right=of Lone0](Lone1){$q,A$};
    \path (Lone0) edge node{$-\vec e_{\mathbf 1}$} (Lone1);
    \node[state,right=4.82cm of Lone1,label=left:{R$_{\mathbf 1}$:}](Rone0){$\emptyset,\mathbf 1$};
    \node[state,accepting by double,right=of Rone0](Rone1){$q_\ell$};
    \path (Rone0) edge node{$\vec 0$} (Rone1);
    \node[state,below=1cm of Lone0,label=left:{L$_\oplus$:}](Loplus0){$q,C$};
    \node[intermediate,right=1.1cm of
      Loplus0,label=right:{$\wedge$}](Loplus1){};
    \draw[draw=black!40] (Loplus1) -- ++(.2,-.2) arc (-40:40:.33);
    \node[intermediate,above right=of Loplus1](Loplus2){};
    \node[intermediate,below right=of Loplus1](Loplus3){};
    \node[state,right=1cm of Loplus1](Loplus4){$q,C$};
    \path (Loplus0) edge node{$-\vec e_{A\oplus B}$} (Loplus1)
      (Loplus1) edge (Loplus2)
      (Loplus1) edge (Loplus3)
      (Loplus2) edge node{$\vec e_A$} (Loplus4)
      (Loplus3) edge[swap] node{$\vec e_B$} (Loplus4);
    \node[rstate,right=2.95cm of
      Loplus4,label=left:{R$_\oplus$:}](Roplus0){$q,A\oplus B$};
    \node[state,above right=.15cm and .3cm of Roplus0](Roplus1){$q,A$};
    \node[state,below right=.15cm and .3cm of Roplus0](Roplus2){$q,B$};
    \path (Roplus0.east) edge node {$\vec 0$} (Roplus1)
      (Roplus0.east) edge[swap] node {$\vec 0$} (Roplus2);
    \node[state,below left=2cm and -.6cm of
      Loplus0,label=left:{L$_\with$:}](Lwith0){$q,C$};
    \node[intermediate,right=1cm of Lwith0](Lwith1){};
    \node[state,above right=of Lwith1](Lwith2){$q,C$};
    \node[state,below right=of Lwith1](Lwith3){$q,C$};
    \path (Lwith0) edge node{$-\vec e_{A\with B}$} (Lwith1)
      (Lwith1) edge node{$\vec e_A$} (Lwith2)
      (Lwith1) edge[swap] node{$\vec e_B$} (Lwith3);
    \node[rstate,right=4.82cm of
      Lwith1,label=left:{R$_\with$:},label=right:{$\,\wedge$}](Rwith0){$q,A\with B$};
    \draw[draw=black!40] (Rwith0.east) -- ++(.1,-.1) arc (-40:40:.15);
    \node[state,above right=.15cm and .3cm of Rwith0](Rwith1){$q,A$};
    \node[state,below right=.15cm and .3cm of Rwith0](Rwith2){$q,B$};
    \path (Rwith0.east) edge (Rwith1)
      (Rwith0.east) edge (Rwith2);
    \node[state,below right=1cm and 3.5cm of
      Lwith0,label=left:{R$_\top$:}](Rtop0){$q,\top$};
    \node[intermediate,right=of Rtop0,label=below:{{\footnotesize $\forall A\in
  S\setminus S_\oc$}}](Rtop1){};
    \node[state,accepting by double,right=of
      Rtop1](Rtop2){$q_\ell$};
    \path (Rtop0) edge node{$\vec 0$} (Rtop1)
      (Rtop1) edge[loop above] node {$-\vec e_A$} ()
      (Rtop1) edge node{$\vec 0$} (Rtop2);
    \node[rstate,below left=2.2cm and -1.3cm of
      Lwith0,label=left:{$\oc$D:}](D0){$q\cup\{\oc A\},B$};
    \node[state,right=of D0](D1){$q,B$};
    \path (D0) edge node{$\vec e_A$} (D1);
    \node[rstate,right=1.2cm of D1,label=left:{$\oc$W:}](W0){$q\cup\{\oc A\},B$};
    \node[state,right=of W0](W1){$q,B$};
    \path (W0) edge node{$\vec 0$} (W1);
    \node[state,right=1.2cm of W1,label=left:{$\oc$P:}](P0){$q,\oc A$};
    \node[state,right=of P0](P1){$q,A$};
    \path (P0) edge node {$\stackrel{?}{=}\vec 0$} (P1);
  \end{tikzpicture}
  \caption{\label{fig-ill-abvass}The rules of $\?A^I_F$; $q,q'$ are
    subsets of $S_\oc$, all the formul\ae\ must be in $S$.}
\end{figure}

For $A$ empty or a formula in $S$, we write $A^\dagger\eqdef .$ if $A$
is empty and $A^\dagger\eqdef A$ otherwise.  The following claim
relates ILZ proofs with deductions in $\?A^I_F$: \stepcounter{theorem}
\begin{claim}[$\?A^I_F$ is Sound and Complete]\label{cl-ill}
  For all $\oc\Psi$ in $\+N^{S_\oc}$, $\Delta$ in $\+N^{S\setminus
    S_\oc}$ and $A^\dagger$ in $S\uplus\{.\}$,
  \begin{equation*}
    \?A^I_F,\{q_\ell\}\jdg \sigma(\oc\Psi),A^\dagger,\Delta
    \quad\text{iff}\quad
    \oc\Psi,\Delta\prvs{ILZ} A\;.
  \end{equation*}
\end{claim}
\begin{proof}[Completeness Proof]
  Let us prove by induction on the height of a proof tree for
  $\oc\Psi,\Delta\prvs{ILZ} A$ that $\?A^I_F,\{q_\ell\}\jdg
  \sigma(\oc\Psi),A^\dagger,\Delta$.  This boils down to a
  verification that the rules in \autoref{fig-ill-abvass} implement
  the sequent calculus for ILZ.  We will not detail all the cases, but
  here are four instances for (init), (L$_\multimap$), (L$_\oplus$), and
  ($\oc$P)---the remaining cases being similar.
  {\ifacm\def\descriptionlabel#1{%
    \hspace\labelsep \normalfont\itshape #1,%
    }%
    \fi
  \begin{description}
  \item[For (init)] we know that
    $\?A_F^I,\{q_\ell\}\jdg q_\ell,\vec 0$
  vacuously.  We then distinguish two cases: either $A$ is not in
  $S_\oc$, and the first (init) unary rule applies and allows to
  deduce
    $\?A_F^I,\{q_\ell\}\jdg \emptyset,A,\vec e_A$
  as desired, or $A$ is in $S_\oc$, and the second (init) unary rule
  applies and allows to deduce 
    $\?A_F^I,\{q_\ell\}\jdg \{A\},A,\vec 0$
  as desired.
  \item[For (L$_\multimap$)] let us write $\oc\Phi,\Gamma\prvs{ILZ} A$ and
  $\oc\Psi,\Delta,B\prvs{ILZ}C$ for the premises with $\oc\Phi$ and
  $\oc\Psi$ in $\+N^{S_\oc}$, $\Gamma$ and $\Delta$ in
  $\+N^{S\setminus S_\oc}$, and $A,B,C$ in $S$.  By induction
  hypothesis,
    $\?A_F^I,\{q_\ell\}\jdg \sigma(\oc\Phi),A,\Gamma$.
  If $B$ is in $S_\oc$, then by induction
  hypothesis
    $\?A_F^I,\{q_\ell\}\jdg \sigma(\oc\Psi)\cup B,C,\Delta$
  and we can apply the store rule to show
    $\?A_F^I,\{q_\ell\}\jdg \sigma(\oc\Phi),C,\Delta+\vec e_B$.
  Otherwise, i.e.\ if $B$ is in $S\setminus S_\oc$, the induction
  hypothesis provides us with the same sequent directly.  %
  Applying the rules for (L$_\multimap$), we see that
    $\?A_F^I,\{q_\ell\}\jdg\sigma(\oc\Phi)\cup\sigma(\oc\Psi),C,\Gamma+\Delta+\vec e_{A\multimap B}$
  as desired.
  \item[For (L$_\oplus$)] let us write $\oc\Phi,\Gamma,A\prvs{ILZ} C$ and
  $\oc\Phi,\Gamma,B\prvs{ILZ}C$ for the premises with $\oc\Phi$ in
  $\+N^{S_\oc}$, $\Gamma$ in $\+N^{S\setminus S_\oc}$, and $A,B,C$ in
  $S$.  Like in the previous argument for (L$_\multimap$), we can use
  the store rule if $A$ or $B$ are in $S_\oc$, and in any case the
  induction hypothesis shows
  $\?A_F^I,\{q_\ell\}\jdg \sigma(\oc\Phi),C,\Gamma+\vec{e}_A$ and
  $\?A_F^I,\{q_\ell\}\jdg \sigma(\oc\Phi),C,\Gamma+\vec{e}_B$.
  The \abv\ rules for (L$_\oplus$) then allow to remove $\vec e_A$ and
  $\vec e_B$ from these two configurations, apply the fork rule, and
  finally add $\vec e_{A\oplus B}$ to prove
  $\?A_F^I,\{q_\ell\}\jdg \sigma(\oc\Phi),C,\Gamma+\vec{e}_{A\oplus
  B}$ as desired.
  \item[For ($\oc$P)] by induction hypothesis
    $\?A_F^I,\{q_\ell\}\jdg \sigma(\oc\Gamma),A,\vec 0$
  and we can apply the corresponding full zero test to show
    $\?A_F^I,\{q_\ell\}\jdg \sigma(\oc\Gamma),\oc A,\vec 0$
  as desired.\qedhere
  \end{description}}
\end{proof}
\begin{proof}[Soundness Proof]
  As a preliminary observation, note that the store rule of
  $\?A_F^I$ is the only rule that can decrement the counter of a
  formula from $S_\oc$.  By induction over the height of deduction
  trees for $\?A_F^I$, we can normalise deductions so that store
  rules are applied either immediately after a rule that added $\vec
  e_{\oc A}$ to the current configuration, or immediately at the root
  of the deduction (if $F$ itself is in $S_\oc$).  This means that we
  can assume $\Delta$ in $\+N^{S\setminus S_\oc}$ in a root judgement
  $\?A^I_F,\{q_\ell\}\jdg \sigma(\oc\Psi),A^\dagger,\Delta$ as in the
  statement of the claim.

  Let us prove by induction on the height of a deduction tree for
  $\?A^I_F,\{q_\ell\}\jdg \sigma(\oc\Psi),A^\dagger,\Delta$ that
  $\oc\Psi,\Delta\prvs{ILZ} A$.  Again, we only cover a few cases,
  depending on which group of rules of $\?A_F^I$ was applied last:
  {\ifacm\def\descriptionlabel#1{%
    \hspace\labelsep \normalfont\itshape #1,%
    }%
    \fi
  \begin{description}
  \item[For (init)] either the first variant for $A$ in $S\setminus S_\oc$
  was used to obtain $\?A_F^I,\{q_\ell\}\jdg \emptyset,A,\vec e_A$,
  and we have $A\prvs{ILZ} A$ as desired, or the second variant for
  $\oc A$ in $S_\oc$ was used to obtain $\?A_F^I,\{q_\ell\}\jdg \oc
  A,\oc A,\vec 0$, and we have $\oc A\prvs{ILZ}\oc A$ as desired.
  \item[For (L$_\multimap$)] we know that $\?A_F^I,\{q_\ell\}\jdg
  q,A,\Gamma$ and $\?A_F^I,\{q_\ell\}\jdg q',C,\Delta+\vec e_B$ for
  some $q$ and $q'$ included in $S_\oc$, some $\Gamma$ and $\Delta$ in
  $\+N^{S\setminus S_\oc}$, and some $A$, $B$, and $C$ in $S$.  By
  induction hypothesis, there exists $\oc\Phi$ in $\+N^{S_\oc}$ with
  $q=\sigma(\oc\Phi)$ such that $\oc\Phi,\Gamma\prvs{ILZ}A$.  If $B$
  is in $S_\oc$---and is therefore the result of a store rule applied
  right after---, then $\?A_F^I,\{q_\ell\}\jdg q'\cup\{B\},C,\Delta$
  and by induction hypothesis there exists $\oc\Psi$ in $\+N^{S_\oc}$
  with $\sigma(\oc\Psi)=q'$ such that $\oc\Psi,B,\Delta\prvs{ILZ}C$.
  Otherwise, i.e.\ if $B$ is in $S\setminus S_\oc$, we obtain the
  latter sequent directly by induction hypothesis.  Then
  $\oc\Psi,\oc\Phi,B,\Delta,\Gamma,A\multimap B\prvs{ILZ}C$ as
  desired.
  \item[For (L$_\oplus$)] we know that $\?A_F^I,\{q_\ell\}\jdg
  q,C,\Gamma+\vec e_A$ and $\?A_F^I,\{q_\ell\}\jdg q,C,\Gamma+\vec
  e_B$ for some $q$ included in $S_\oc$, some $\Gamma$ in
  $\+N^{S\setminus S_\oc}$, and some $A$, $B$, and $C$ in $S$.  Using
  store rules, we can argue as in the case of (L$_\multimap$) that,
  regardless of whether $A$ or $B$ is in $S_\oc$, by induction
  hypothesis there exists $\oc\Phi$ in $\+N^{S_\oc}$ with
  $q=\sigma(\oc\Phi)$ such that $\oc\Phi,A,\Gamma\prvs{ILZ}C$ and
  $\oc\Phi,B,\Gamma\prvs{ILZ}C$.  Therefore $\oc\Phi,A\oplus
  B,\Gamma\prvs{ILZ}C$ as desired.
  \item[For ($\oc$P)] we know that $\?A_F^I,\{q_\ell\}\jdg q,A,\vec 0$,
  thus by induction hypothesis there exists $\oc\Psi$ in $\+N^{S_\oc}$
  with $q=\sigma(\oc\Psi)$ and $\oc\Psi\prvs{ILZ}A$, hence
  $\oc\Psi\prvs{ILZ}\oc A$ as desired.\qedhere
  \end{description}}
\end{proof}

\begin{claim}[Affine Case]\label{cl-illw}
  When allowing losses in $\?A_F^I$, for all $\oc\Psi$ in
  $\+N^{S_\oc}$, $\Delta$ in $\+N^{S\setminus S_\oc}$ and $A^\dagger$
  in $S\uplus\{.\}$,
  \begin{equation*}
    \?A^I_F,\{q_\ell\}\jdg_\ell \sigma(\oc\Psi),A^\dagger,\Delta
    \quad\text{iff}\quad
    \oc\Psi,\Delta\prvs{ILZW} A\;.
  \end{equation*}
\end{claim}
\begin{proof}
  As a preliminary observation, note that, by monotonicity, losses
  occurring inside a group of rules depicted in
  \autoref{fig-ill-abvass} can be delayed until after the execution of
  the group is completed.  By \autoref{cl-ill}, it therefore suffices
  to check the case of the loss and structural weakening rules.
  
  For completeness, if (W) is the last applied rule in a proof of
  $\oc\Psi,B,\Delta\prvs{ILZW} A$, then, by induction hypothesis
  $\?A_F^I,\{q_\ell\}\jdg_\ell\sigma(\oc\Psi),A,\Delta$, and an
  application of the loss rule yields
  $\?A_F^I,\{q_\ell\}\jdg_\ell\sigma(\oc\Psi),A,\Delta+\vec e_B$ as
  desired.

  For soundness, if a loss of some $B$ is the last applied rule in a
  deduction showing $\?A_F^I,\{q_\ell\}\jdg_\ell q,A,\Delta+\vec e_B$,
  then $\?A_F^I,\{q_\ell\}\jdg_\ell q,A,\Delta$.  By induction hypothesis
  there exists $\oc\Psi$ with $q=\sigma(\oc\Psi)$ such that
  $\oc\Psi\Delta\prvs{ILZW}A$, from which (W) yields
  $\oc\Psi,B,\Delta\prvs{ILZW}A$ as desired.
\end{proof}

\begin{claim}[Contractive Case]\label{cl-illc}
  When allowing expansions in $\?A_F^I$, for all $\oc\Psi$ in
  $\+N^{S_\oc}$, $\Delta$ in $\+N^{S\setminus S_\oc}$ and $A^\dagger$
  in $S\uplus\{.\}$,
  \begin{equation*}
    \?A^I_F,\{q_\ell\}\jdg_e \sigma(\oc\Psi),A^\dagger,\Delta
    \quad\text{iff}\quad
    \oc\Psi,\Delta\prvs{ILZC} A\;.
  \end{equation*}
\end{claim}
\begin{proof}
  As a preliminary observation, note that, by monotonicity, expansions
  occurring inside a group of rules depicted in
  \autoref{fig-ill-abvass} can be applied before the execution of the
  group is started.  By \autoref{cl-ill}, it therefore suffices to
  check the case of the expansion and structural contraction rules.

  For completeness, if (C) is the last applied rule in a proof of
  $\oc\Psi,B,\Delta\prvs{ILZC} A$, then we can assume that the
  contracted formula was some $B$ in $S\setminus S_\oc$ as otherwise
  logical weakening would have sufficed, thus
  $\oc\Psi,B,B,\Delta\prvs{ILZW}A$.  By induction hypothesis,
  $\?A_F^I,\{q_\ell\}\jdg_e\sigma(\oc\Psi),A,\Delta+2\vec e_B$, and an
  application of the expansion deduction rule yields
  $\?A_F^I,\{q_\ell\}\jdg_e\sigma(\oc\Psi),A,\Delta$ as desired.

  For soundness, assume that an expansion is the rule applied at the
  root of a deduction tree for $\?A_F^I,\{q_\ell\}\jdg_e
  q,A,\Delta+\vec e_B$, hence that $\?A_F^I,\{q_\ell\}\jdg_e
  q,A,\Delta+2\vec e_B$.  Because we assume store rules to occur as
  early as possible, $B$ cannot be in $S_\oc$.  Thus, by induction
  hypothesis there exists $\oc\Psi$ with $q=\sigma(\oc\Psi)$ and
  $\oc\Psi,B,B,\Delta\prvs{ILZC} A$, and applying (C) yields
  $\oc\Psi,B,\Delta\prvs{ILZC}A$ as desired.
\end{proof}

\addtocounter{theorem}{-1}
\begin{proposition}\label{prop-ill-abvass}
  There are polynomial space reductions:
  \begin{enumerate}
  \item\label{ill-1} from (affine, resp.\ contractive) ILZ provability
    to (lossy, resp.\ expansive) \abv\ reachability,
  \item\label{ill-2} from (affine, resp.\ contractive) IMELZ
    provability to (lossy, resp.\ expansive) BVASS$_{\vec 0}$
    reachability.
  \end{enumerate}
\end{proposition}
\begin{proof}
  For~\ref{ill-1}, we reduce the provability of $\vdash F$ to the
  reachability of $(\emptyset,F)$ in $\?A_F^I$, which is correct
  thanks to the subformula property and
  claims~\ref{cl-ill}--\ref{cl-illc}.

  For~\ref{ill-2}, simply observe that (L$_\oplus$) and (R$_\with$)
  are the only rules of ILZ that make use of fork rules in $\?A_F^I$.
\end{proof}

%% file: sec-abvass-ll.tex
\subsection{From \abv\ to LL}
In order to exhibit a reduction from \abv\ reachability to LL
provability, we extend a similar reduction proved by
\citet*{lincoln92} in the case of AVASS (also employed by
\citet{urquhart99}).  The general idea is to encode \abv\
configurations as sequents and \abv\ deductions as proofs in LL
extended with a \emph{theory}, where encoded \abv\ rules are provided
as an additional set of non-logical axioms.

\subsubsection{Linear Logic with a Theory\nopunct.}
In the framework of \citeauthor{lincoln92}, a theory $T$ is a finite
set of \emph{axioms} $C,p_1^\bot,\dots,p_m^\bot$ where $C$ is a MALL
formula and each $p_i$ is an atomic proposition.  Proofs in LL$+T$ can
employ two new rules
\begin{equation*}
  \drule{}{\vdash C,p_1^\bot,\dots,p_m^\bot}{$T$}\quad
  \drule{\vdash C,p_1^\bot,\dots,p_m^\bot\quad\vdash
    C^\bot,\Delta}{\vdash p_1^\bot,\dots,p_m^\bot,\Delta}{directed cut}
\end{equation*}
where $C,p_1^\bot,\dots,p_m^\bot$ belongs to $T$.

A proof in LL$+T$ is directed if all its cuts are \emph{directed
  cuts}.  By adapting the LL cut-elimination proof,
\citeauthor{lincoln92} show:
\begin{fact}[\citep{lincoln92}]\label{fc-dcut}
  If there is a proof of $\vdash\Gamma$ in LL$+T$, then there is a
  directed proof of $\vdash\Gamma$ in LL$+T$.
\end{fact}%

The axioms of a theory $T$ can be translated in pure LL by
\begin{align}
  \enc{C,p_1^\bot,\dots,p_m^\bot}&\eqdef C^\bot\otimes
  p_1\otimes\cdots\otimes p_m\;.
\end{align}\vspace*{-1.5em}
\begin{fact}[\citep{lincoln92}]\label{fc-llt}
  For any finite set of axioms $T$, $\vdash\Gamma$ is provable in
  LL$+T$ if and only if $\vdash\wn\enc{T},\Gamma$ is provable in LL.
\end{fact}

\subsubsection{Encoding \abv\nopunct.}
Given an ABVASS $\?A=\tup{Q,d,T_u,T_f,T_s,\emptyset}$, a configuration
$(q,\vec v)$ in $Q\times\+N^d$ is encoded as the sequent
\begin{equation}
  \theta(q,\vec v)\eqdef\:\vdash q^\bot,(e_1^\bot)^{\vec
    v(1)},\dots,(e_d^\bot)^{\vec v(d)}
\end{equation}
where $Q\uplus\{e_i\mid i=1,\dots,d\}$ is included in the set of
atomic propositions and $A^n$ stands for the formula $A$ repeated $n$
times.

By \autoref{lem-ordinary} we assume $\?A$ to be in ordinary form.  We
construct from the rules of $\?A$ a theory $T$ consisting of sequents
of form $\vdash q^\bot,c_1^\bot,\dots,c_m^\bot,C$ with $q$ in $Q$ the
originating state, $c_j$ in $\{e_i\mid0<i\leq d\}$, and $C$ a MALL
formula containing the destination state(s) positively.  Here are the
axioms corresponding to each type of rule:
\begin{align*}
q&\xrightarrow{\vec e_i} q_1
 &&q^\bot,q_1\otimes e_i\\
q&\xrightarrow{-\vec e_i} q_1
 &&q^\bot,e_i^\bot,q_1\\
q&\to q_1\wedge q_2
 &&q^\bot,q_1\oplus q_2\\
q&\to q_1+q_2
 &&q^\bot,q_1\parr q_2
 \shortintertext{By \autoref{lem-abvass}, we do not need to consider the
   case of full zero tests.  Here is nevertheless how they could be
   encoded, provided we slightly extended the reduction
   of LL$+T$ to LL in \autoref{fc-llt} to allow exponentials in $T$:}
q&\rst q_1
 &&q^\bot,\oc q_1
\end{align*}

\stepcounter{theorem}
\begin{claim}\label{cl-abvass-ll}
  For all $(q,\vec v)$ in $Q\times\+N^d$, $\?A,Q_\ell\jdg q,\vec v$
  if and only if $\vdash\theta(q,\vec v),\wn Q_\ell$ in LL$+T$.
\end{claim}
\begin{proof}
  The AVASS case is proved by \citet[Lemmata~3.5 and~3.6]{lincoln92}
  by induction on the height of deduction trees in $\?A$ and the
  number of directed cuts in a directed proof in LL$+T$ (with minor
  adaptations for $\wn Q_\ell$).  Thus, we only need to prove that
  split rules preserve this statement.\footnote{\Citet{degroote04}
    show how to handle split rules in IMELL, but they do not rely on
    the LL$+T$ framework, which motivates considering this case here.}

  Assume for the direct implication that $\?A,Q_\ell\jdg q,\vec v$
  as the result of a split rule $q\to q_1+q_2$, thus $\vec v=\vec
  v_1+\vec v_2$ and $\?A,Q_\ell\jdg q_1,\vec v_1$ and
  $\?A,Q_\ell\jdg q_2,\vec v_2$.  By induction hypothesis,
  $\vdash\theta(q_1,\vec v_1),\wn Q_\ell$ and $\vdash\theta(q_2,\vec
  v_2),\wn Q_\ell$, and we have the following proof of
  $\vdash\theta(q,\vec v),\wn Q_\ell$:  \vspace*{-1.5em}
  {\begin{prooftree}%
    \AxiomC{i.h.}
    \UnaryInfC{$\vdash q_1^\bot,(c_1^\bot)^{\vec
        v_1(1)},\ldots,(c_d^\bot)^{\vec v_1(d)},\wn Q_\ell$}
    \AxiomC{i.h.}
    \UnaryInfC{$\vdash q_2^\bot,(c_1^\bot)^{\vec
        v_2(1)},\ldots,(c_d^\bot)^{\vec v_2(d)},\wn Q_\ell$}
    \LeftLabel{{\tiny($\otimes$)}}
    \BinaryInfC{$\vdash q_1^\bot\otimes
      q_2^\bot,(c_1^\bot)^{\vec v_1(1)+\vec v_2(1)},\dots,(c_d^\bot)^{\vec
        v_1(d)+\vec v_2(d)},\wn Q_\ell,\wn Q_\ell$}
    \doubleLine
    \LeftLabel{{\tiny({?}C)}}
    \UnaryInfC{$\vdash q_1^\bot\otimes
      q_2^\bot,(c_1^\bot)^{\vec v_1(1)+\vec v_2(1)},\dots,(c_d^\bot)^{\vec
        v_1(d)+\vec v_2(d)},\wn Q_\ell$}
    \AxiomC{$\!\!\!\!\!\vdash q^\bot, q_1\parr q_2$}
    \LeftLabel{{\tiny (dir.\ cut)}}
    \BinaryInfC{$\vdash q^\bot,(c_1^\bot)^{\vec v_1(1)+\vec v_2(1)},\dots,(c_d^\bot)^{\vec
        v_1(d)+\vec v_2(d)},\wn Q_\ell$}
  \end{prooftree}}

  Conversely, assume that the last applied directed cut has
  \begin{equation}\label{eq-ll-split}\mbox{$\vdash q_1^\bot\otimes
      q_2^\bot$},(c_1^\bot)^{\vec v(1)},\dots,(c_d^\bot)^{\vec v(d)},\wn
    Q_\ell
  \end{equation} and $\vdash
  q^\bot,q_1\parr q_2$ as premises.  The only rules that allow to
  prove \eqref{eq-ll-split} are ($\wn$D), ($\wn$C) or ($\wn$W) applied to some $q_\ell$
  in $Q_\ell$, and ($\otimes$).  Logical contractions are irrelevant
  to the claim, and wlog.\ we can apply derelictions above
  ($\otimes$), thus we know that \eqref{eq-ll-split} is the result of
  ($\otimes$) followed by a series of ($\wn$W).  Hence $\vdash\theta(q_1,\vec v_1),\wn Q_1$ and $\vdash\theta(q_2,\vec v_2),\wn Q_2$ with $\vec v=\vec v_1+\vec
  v_2$ and $Q_\ell\supseteq Q_1\cup Q_2$.  By induction hypothesis,
  $\?A,Q_1\jdg q_1,\vec v_1$ and $\?A,Q_2\jdg q_2,\vec v_2$.  Because
  $Q_1\subseteq Q_\ell$ and $Q_2\subseteq Q_\ell$ this entails
  $\?A,Q_\ell\jdg q_1,\vec v_1$ and $\?A,Q_\ell\jdg q_2,\vec v_2$,
  from which a split allows to derive $\?A,Q_\ell\jdg q,\vec v$ as
  desired.
\end{proof}
\addtocounter{theorem}{-1}
\begin{proposition}\label{prop-abvass-ll}
  There are logarithmic space reductions
  \begin{enumerate}
    \item\label{abvass-ll1} from \abv\ reachability to LL provability and
    \item\label{abvass-ll2} from BVASS$_{\vec 0}$ reachability to MELL
      provability.
  \end{enumerate}
\end{proposition}
\begin{proof}
  By \autoref{lem-abvass} we can eliminate full zero tests.
  For~\ref{abvass-ll1}, by \autoref{cl-abvass-ll} and
  \autoref{fc-llt}, $\?A,Q_\ell\jdg q_r,\vec 0$ if and only if
  \mbox{$\vdash q_r^\bot,\wn Q_\ell,\wn\enc{T}$}.
  Regarding~\ref{abvass-ll2}, simply observe that additive connectives
  are only used for the encoding of fork rules.
\end{proof}
\subsubsection{Affine Case\nopunct.}
Adapting the proof of \autoref{prop-abvass-ll} to the affine case is
relatively straightforward.  For starters, \autoref{fc-dcut} also
holds for LLW$+T$ using the cut elimination procedure for LLW, and
allowing structural weakenings does not influence the proof of
\autoref{fc-llt} in \citep[Lemmata~3.2 and~3.3]{lincoln92}.  We show:
\begin{proposition}\label{prop-abvass-llw}
  There are logarithmic space reductions
  \begin{enumerate}
  \item\label{abvass-llw1} from \abv\ lossy reachability to LLW
    provability and
  \item\label{abvass-llw2} from BVASS$_{\vec 0}$ lossy reachability to
    MELLW provability.
  \end{enumerate}
\end{proposition}\noindent
This relies on an extension of \autoref{cl-abvass-ll}%
\ifsubmission%
\ proved in \appref{sec-omitted}%
\fi%
:%
\begin{restatable}{claim}{clabvassllw}\label{cl-abvass-llw}
  For all $(q,\vec v)$ in $Q\times\+N^d$, $\?A,Q_\ell\jdg_\ell q,\vec v$
  with lossy semantics if and only if $\vdash\theta(q,\vec v),\wn
  Q_\ell$ in LLW$+T$.
\end{restatable}%
\ifomitproofs\else
\begin{proof}
  \input{proof-cl-abvass-llw}
\end{proof}
\fi%

\subsubsection{Contractive Case\nopunct.}
Again, \autoref{fc-dcut} is straightforward to adapt to LLC$+T$ using
cut elimination.  \autoref{fc-llt} can be strengthened to avoid
exponentials in the contractive case%
%
:
\begin{restatable}{lemma}{lemllct}\label{lem-llct}
  For a finite set of axioms $T$, $\vdash\Gamma$ is provable in
  LLC$+T$ %
  if and only if $\vdash\bot\oplus\bigoplus_{t\in
    T}\enc{t},\Gamma$ is provable in LLC.
\end{restatable}
\begin{proof}
  \input{proof-lem-llct}
\end{proof}

Without loss of generality, we can assume that $Q_\ell=\{q_\ell\}$ for
a state $q_\ell$ with no applicable rule in $\?A$.  We extend
\autoref{cl-abvass-ll} and \autoref{prop-abvass-ll} to the contractive case%
%
:
\stepcounter{theorem}
\begin{restatable}{claim}{clabvassllc}\label{cl-abvass-llc}
  For all $(q,\vec v)$ in $Q\times\+N^d$, $\?A,\{q_\ell\}\jdg_e q,\vec
  v$ using expansive semantics if and only if $\vdash\theta(q,\vec
  v),q_\ell^s$ in LLC$+T$ for some $s>0$.
\end{restatable}
\ifomitproofs\else
\begin{proof}
  \input{proof-cl-abvass-llc}
\end{proof}
\fi
\addtocounter{theorem}{-1}\begin{restatable}{proposition}{propabvassllc}\label{prop-abvass-llc}
  There is a logarithmic space reduction from \abv\ expansive
  reachability to MALLC provability.
\end{restatable}
\ifomitproofs\else\input{proof-prop-abvass-llc}\fi

%% file: proof-cl-abvass-llw.tex
We proceed as before by induction on the height of a deduction tree
in $\?A$ and on the number of directed cuts in a proof in LLW$+T$.
The only new cases to consider in addition to those of
\autoref{cl-abvass-ll} are those of losses and structural weakenings.
In case of a loss allowing to derive $\?A,Q_\ell\jdg_\ell q,\vec v+\vec
e_i$, by induction hypothesis $\vdash\theta(q,\vec v),\wn Q_\ell$ and
a structural weakening yields $\vdash\theta(q,\vec v+\vec e_i),\wn
Q_\ell$ as desired.  Conversely, in case of a structural weakening
allowing to derive $\vdash\theta(q,\vec v),\wn Q_\ell$, either the
weakened formula is some $\wn q_\ell$ guarded by $\wn$ and by
induction hypothesis $\?A,Q_\ell\setminus\{q_\ell\}\jdg_\ell q,\vec v$
thus $\?A,Q_\ell\jdg_\ell q,\vec v$ as desired, or the weakened formula
is some $c_i^\bot$ and by induction hypothesis $\?A,Q_\ell\jdg_\ell
(q,\vec v-\vec e_i)$ and a loss allows to derive $\?A,Q_\ell\jdg_\ell
q,\vec v$ as desired.

%% file: proof-lem-llct.tex
For the direct implication, we consider a directed proof of
$\prvs{LLC$+T$}\Gamma$.  By induction on the number of directed cuts,
we build an LLC proof of $\prvs{LLC}\bot\oplus\bigoplus_{t\in
T}\enc{t},\Gamma$.  For the base case, an LLC$+T$ proof of
$\prvs{LLC$+T$}\Gamma$ without directed cuts is also an LLC proof of
$\prvs{LLC}\Gamma$, thus $\prvs{LLC}\bot,\Gamma$ using the ($\bot$)
rule, and $\vdash\bot\oplus\bigoplus_{t\in T}\enc{t},\Gamma$ by $|T|$
applications of ($\oplus$).  For the induction step, consider a
directed cut of an axiom $t=C,p_1^\bot,\dots,p_m^\bot$ in $T$ with
$\prvs{LLC$+T$} C^\bot,\Delta$.  We have $\prvs{LLC} C,C^\bot$ and
$\prvs{LLC} p_i,p_i^\bot$ for all $0<i\leq m$ by the (init) rule, and
$m+1$ applications of ($\otimes$) yield $\prvs{LLC} t,\enc{t}$.  By
induction hypothesis $\prvs{LLC}\bot\oplus\bigoplus_{t\in
T}\enc{t},C^\bot,\Delta$, thus a (normal) cut yields
$\prvs{LLC}\bot\oplus\bigoplus_{t\in
T}\enc{t},\enc{t},p_1^\bot,\dots,p_m^\bot,\Delta$.  Using $|T|$
applications of ($\oplus$) allows to prove
$\prvs{LLC}\bot\oplus\bigoplus_{t\in
T}\enc{t},\bot\oplus\bigoplus_{t\in
T}\enc{t},p_1^\bot,\dots,p_m^\bot,\Delta$ and a structural contraction
yields the desired LLC proof.

For the converse implication, if $\prvs{LLC}\bot\oplus\bigoplus_{t\in
  T}\enc{t},\Gamma$, then $\prvs{LLC$+T$}\bot\oplus\bigoplus_{t\in
  T}\enc{t},\Gamma$.  Then $\prvs{LLC$+T$}\mathbf 1$, and for each
axiom $t=C,p_1^\bot,\dots,p_m^\bot$ in $T$, we can prove
$\prvs{LLC$+T$}C\parr p_1^\bot\parr\cdots\parr p_m^\bot$ by $m$
applications of ($\parr$) from $\prvs{LLC$+T$}t$,
i.e.\ $\prvs{LLC$+T$}\enc{t}^\bot$.  Thus $|T|$ applications of
($\with$) yield $\prvs{LLC$+T$}\mathbf 1\with\bigwith_{t\in
  T}\enc{t}^\bot$, and a (normal) cut shows $\prvs{LLC$+T$}\Gamma$.

%% file: proof-cl-abvass-llc.tex
By \autoref{cl-abvass-ll}, it suffices to consider the case of
expansions and structural contractions in a proof by induction over
deduction tree height and number of directed cuts.  In case of an
expansion allowing to derive $\?A,\{q_\ell\}\jdg_e q,\vec v+\vec e_i$,
by induction hypothesis, $\vdash\theta(q,\vec v+2\vec e_i),q_\ell^s$
and a structural contraction allows to prove $\vdash\theta(q,\vec
v+\vec e_i),q_\ell^s$ as desired.  Conversely, in case of a structural
contraction proving \mbox{$\vdash\theta(q,\vec v),q_\ell^s$}, several
cases are possible.  If the contracted formula is $q_\ell$, then by
induction hypothesis $\?A,\{q_\ell\}\jdg_e q,\vec v$ as desired.  If
the contracted formula is some $e_i^\bot$ with $0<i\leq d$, then by
induction hypothesis $\?A,\{q_\ell\}\jdg_e q,\vec v+2\vec e_i$ and an
expansion allows to deduce $\?A,\{q_\ell\}\jdg_e q,\vec v+\vec e_i$ as
desired.  Last of all, the contracted formula cannot be $q^\bot$:
assume for the sake of contradiction that $\vdash
q_1^\bot,\dots,q_n^\bot,\theta(q,\vec v),q_\ell$ were provable in
LLC$+T$ for some $n>0$ negated atomic state propositions
$q_1^\bot,\dots,q_n^\bot$ (in addition to $q^\bot$), and attempt to
perform directed proof search.  The only applicable rules are
\begin{itemize}
\item structural contraction, which cannot decrease $n$, and
\item directed cuts using $T$, which also preserve $n$.
\end{itemize}
In the absence of any axiom allowing $n>0$, this sequent is not
provable.

%% file: proof-prop-abvass-llc.tex
\begin{proof}
  As usual, we start by eliminating full zero tests using
  \autoref{lem-abvass}.  Let $\tup{\?A,q_r,\{q_\ell\}}$ be an
  expansive reachability instance.  By
  \autoref{cl-abvass-llc} and \autoref{lem-llct}
  $\?A,\{q_\ell\}\jdg_e q_r,\vec 0$ if and only if
  $\vdash\bot\oplus\bigoplus_{t\in T}\enc{t},q_r^\bot,q_\ell^s$ for some
  $s>0$, which by structural contractions on $q_\ell$ happens if and
  only if
  $\vdash\bot\oplus\bigoplus_{t\in T}\enc{t},q_r^\bot,q_\ell$.
\end{proof}

%% file: sec-upb.tex
To show that the reachability problem for lossy \abv\ is in \textsc{Tower},
we establish by induction over the dimension $d$ a bound on the height
of minimal reachability witnesses, following in this the reasoning
used by \citet{rackoff78}
to show that the coverability problem for VASS is in \textsc{ExpSpace}.
The main new idea here is that, where there is freedom to choose how
values of vector components are distributed when 
performing split rules top-down (see \autoref{sub-abvass}), 
splitting them equally (or with the difference of $1$) allows sufficient 
lower bounds to be established along vertical paths in deduction trees 
for the inductive argument to go through.
Since the bounds we obtain on the heights of smallest 
witnessing deduction trees are exponentiated at every inductive step 
(rather than multiplied as in \citeauthor{rackoff78}'s proof), 
the resulting complexity upper bound involves a tower of exponentials, 
but will be shown broadly optimal in \autoref{sec-lowb}.

The following lemma in fact addresses the equivalent top-down
coverability problem (see \autoref{s:lossy.reach}), and considers
systems without full resets thanks to \autoref{lem-abvass}.  We first
define some terminology.  %
We say that a deduction tree is:
\begin{itemize}
\item
\emph{$(q_r, \vec v_0)$-rooted} iff 
that is the label of its root;
\item
\emph{$Q_\ell$-leaf-covering} iff, 
for every leaf label $(q, \vec v)$,
we have $q \in Q_\ell$;
\item
of \emph{height} $h$ iff that is the maximum number of edges,
i.e.\ the maximum number of rule applications, 
along any path from the root to a leaf.
\end{itemize}
For integers $d, m \geq 0$ and $s \geq 1$,
we define a natural number $H(d,s, m)$ recursively:
\begin{align}
    H(0, s, m) &\eqdef s\;, \\
H(d + 1,s, m) &\eqdef s (m \cdot 2^{H(d,s, m)})^{d + 1} + H(d,s, m)\;.
\end{align}

\begin{lemma}
\label{l:ABVASS.t.d.c}
If an ABVASS $\?A = \tup{Q,d,T_u,T_f,T_s,\emptyset}$
has a $(q_r, \vec v_0)$-rooted $Q_\ell$-leaf-covering deduction tree,
then it has such a deduction tree of height at most 
$H(d,|Q|, \maxm(T_u))$.
\end{lemma}
\begin{proof}
We use induction on the dimension $d$.

For the base case, if $\?A$ is $0$-dimensional, 
then the labels in its deduction trees are states only.
Starting with a deduction tree whose root label is $q_r$
and whose every leaf label is in $Q_\ell$,
we obtain by repeated shortenings %
a deduction tree in which 
labels along every %
branch are mutually distinct, 
with height at most $|Q| - 1$.

For the induction step in dimension $d+1$, suppose that $\?A =
\tup{Q,d + 1,T_u,T_f,T_s,\emptyset}$, and $\?D$ is a $(q_r, \vec
v_0)$-rooted $Q_\ell$-leaf-covering deduction tree.  Let
\begin{equation}
B\eqdef 2^{H(d,|Q|, \maxm(T_u))} \cdot \maxm(T_u)\;,
\end{equation}
and let $\{n_1, \ldots, n_k\}$ be the set of all nodes of $\?D$ 
such that, for all $i$, we have:
\begin{itemize}
\item
all vector components in labels of ancestors of $n_i$ are smaller than $B$;
\item
for some $0 < j_i \leq d + 1$, we have $\vec v_i(j_i) \geq B$,
where the label of $n_i$ is $(q_i, \vec v_i)$.
\end{itemize}\begin{figure}[tbp]
  \centering
  \begin{tikzpicture}[inner sep=1.5pt,every
    node/.style={font=\footnotesize}]
    \node[label=above:{$q_r,\vec v_0$},draw,circle,fill=black!90](rp){};
    \draw[fill=black!10] (rp) --++ (-1.5,-2) --++ (3,0) -- (rp);
    \node[below left=.75 and -.6 of rp,color=black!60]{$<B$};
    \node[label=above:{$q_1,\vec v_1$},draw,circle,fill=black!90,below
    left=1.7 and .5 of rp](tp1){};
    \draw[fill=white] (tp1) --++ (-.5,-1.3) --++ (1,.1) -- (tp1);
    \node[below=.4 of tp1]{$\?D^\dagger_1$};
    \node[label=above:{$q_k,\vec v_k$},draw,circle,fill=black!90,below
    right=1.7 and .5 of rp](tpn){};
    \draw[fill=white] (tpn) --++ (-.5,-1.3) --++ (1,.1) -- (tpn);
    \node[below=.4 of tpn]{$\?D^\dagger_k$};
    \node[below=2.4 of rp]{$\dots$};
    \path[draw=black!60,sloped,above,every node/.style={font=\tiny}]
      (1.7,0) edge[|-|] node{$\leq |Q|\cdot B^{d+1}$} (1.7,-2)
      (-1.8,-3.2) edge[|-|] node{$\leq H(d,|Q|,\max^-(T_u))$} (-1.8,-1.9)
      (2.7,0) edge[|-|] node{$\leq H(d+1,|Q|,\max^-(T_u))$} (2.7,-3.2);
    \path[draw=black!40,dotted]
      (rp) edge[-] (2.7,0)
      (1.7,-2) --++ (-.3,0)
      (-1.8,-1.9) edge[-] (tp1)
      (-1.8,-3.2) edge[-] (2.7,-3.2);
  \end{tikzpicture}
  \caption{\label{fig-rackoff}Induction step in the proof of \autoref{l:ABVASS.t.d.c}.}
\end{figure}

By repeated shortenings, we can assume that 
the length (i.e., the number of edges) of every path in $\?D$, which is 
from the root either to some $n_i$ or to a leaf with no $n_i$ ancestor, 
is at most $|Q| \cdot B^{d + 1}$, the number of possible labels
with all vector components smaller than $B$.

In the remainder of the argument, we apply the induction hypothesis
below each of the nodes $n_i$.  More precisely, let $\?A_i$ denote the
$d$-dimensional ABVASS obtained from $\?A$ by projecting onto vector
indices $\{1, \ldots, d + 1\} \setminus \{j_i\}$.  (The only change is
in the set of unary rules.)  From the subtree of $\?D$ rooted at
$n_i$, we know that $\?A_i$ has a $(q_i, \vec v_i(-j_i))$-rooted
$Q_\ell$-leaf-covering deduction tree.  (Here $\vec w(-j)$ denotes the
projection of $\vec w$ to all indices except $j$.)  Let $\?D_i$ be
such a deduction tree, which we can choose of height at most $H(d,|Q|,
\maxm(T_u))$ by induction hypothesis.

Now, to turn $\?D_i$ into a $(q_i, \vec v_i)$-rooted 
deduction tree $\?D_i^\dag$ of $\?A$, we have to do two things:
\begin{enumerate}
\item
For every application of a unary rule $q\xrightarrow{\vec u}q'$ in $\?D_i$,
decide which unary rule $q\xrightarrow{\vec {u}'}q'$ of $\?A$
such that $\vec u = \vec u'(-j_i)$ to apply: we do that arbitrarily.
\item
For every application of a split rule $q\to q'+q''$ in $\?D_i$,
decide how to split the vector component $x$ with index $j_i$:
we do that by balancing, i.e.\ picking the corresponding components 
$x_1$ and $x_2$ of the two child vectors so that $|x_1 - x_2| \leq 1$.
\end{enumerate}
We claim that $\?D_i^\dag$ thus obtained is indeed 
a $(q_i, \vec v_i)$-rooted $Q_\ell$-leaf-covering deduction tree of $\?A$.
Since the node labels in $\?D_i^\dag$ differ from those in $\?D_i$
only by the extra $j_i$th components,
it suffices to show that all the latter are non-negative.
In fact, at the root of $\?D_i^\dag$, we have $\vec v_i(j_i) \geq B$,
and it follows by a straightforward induction that, 
for every node $n$ in $\?D_i^\dag$ whose distance from the root is $h$
(which is at most $H(d,|Q|, \maxm(T_u))$),
its vector label $\vec w$ satisfies
\begin{equation}
\vec w(j_i) \geq 2^{H(d,|Q|, \maxm(T_u)) - h} \cdot \maxm(T_u)\;.
\end{equation}

It remains to observe that, by replacing for each $0 < i \leq k$, the
subtree of $\?D$ rooted at $n_i$ by $\?D_i^\dag$, the height of the
resulting deduction tree (see \autoref{fig-rackoff} for a depiction)
is at most
\begin{equation*}
  |Q| \cdot B^{d + 1} + H(d,|Q|, \maxm(T_u)) =
  H(d + 1,|Q|, \maxm(T_u)),
\end{equation*}
thereby establishing the lemma.
\end{proof}

The following auxiliary function and proposition will be useful
for deriving the complexity upper bounds.
Let
\begin{equation}
  H'(d, s, m)\eqdef 4 (d + 1) (s + m + 1) H(d, s, m)\;.
\end{equation}%
%
\begin{restatable}{proposition}{prHp}
\label{pr:Hp}
For all $d, m \geq 0$ and $s \geq 1$, we have:
\[H'(d + 1, s, m) \leq 2^{H'(d, s, m)}\;.\]
\end{restatable}
\ifomitproofs\else
\input{proof-prHp}
\fi

We are now in a position to establish the membership in
\textsc{Tower}.  More precisely, since the height of the tower of
exponentials in the bounds we obtained is equal to the system
dimension, the problem in fixed dimension $d$ is in
$d$-\textsc{ExpTime}.

\begin{theorem}
\label{th:upb}
Reachability for lossy \abv\ is in \textsc{Tower}.
For every fixed dimension $d$, 
it is in \textsc{PTime} if $d = 0$,
and in $d$-\textsc{ExpTime} if $d \geq 1$.
\end{theorem}

\begin{proof}
By \autoref{lem-abvass}, it suffices to consider an ABVASS.
We argue in terms of the
top-down coverability problem
(see \autoref{s:lossy.reach}):
given an ABVASS $\?A = \tup{Q,d,T_u,T_f,T_s,\emptyset}$,
a state $q_r$ and a set of states $Q_\ell$,
to decide whether $\?A$ has a $(q_r, \vec 0)$-rooted $Q_\ell$-leaf-covering 
deduction tree.

By \autoref{l:ABVASS.t.d.c}, if $\?A$ has such a deduction tree,
then it has one of height at most $H(d, |Q|, \maxm(T_u))$.
Observing that, in such a deduction tree, all vector components are bounded by 
\[(\maxp(T_u) + 1) \cdot H(d, |Q|, \maxm(T_u))\;,\]
we conclude that it can be guessed and checked in
\[O((d + 1) \cdot \log ((\maxp(T_u) + 1) \cdot H'(d, |Q|, \maxm(T_u))))\]
space by an alternating algorithm which manipulates 
at most three configurations of $\?A$ at a time.

The memberships in the statement (for ABVASS)
follow from the fact that $H'(0, |Q|, \maxm(T_u))$ is polynomial,
by \autoref{pr:Hp}, and since 
$\mbox{\textsc{ALogSpace}} = \mbox{\textsc{PTime}}$,
$\mbox{\textsc{APSpace}} = \mbox{\textsc{ExpTime}}$, and
$(d - 1)\mbox{-\textsc{AExpSpace}} = d\mbox{-\textsc{ExpTime}}$%
\ifomitproofs\else\ (see \citet{CKS81}).\fi %
\end{proof}

By \autoref{prop-ll-abvass}, this shows:
\begin{corollary}
  LLW provability is in \textsc{Tower}.
\end{corollary}

%% file: proof-prHp.tex
\begin{proof}
We first observe the following inequality involving the $H$ function:
\begin{align*}
  H(d + 1, s, m) &=
  s (m \cdot 2^{H(d, s, m)})^{d + 1} + H(d, s, m)\\
  &\leq s \big((m + 1) \cdot 2^{H(d, s, m)}\big)^{d + 1}\\
  &\leq 2^{(d + 1) (s + m + H(d, s, m))}\\
  &\leq 2^{(d + 1) (s + m + 1) H(d, s, m)}\;,
\end{align*}
and then use it to conclude that:
\begin{align*}
  H'(d + 1, s, m) &=
  4 (d + 2) (s + m + 1) H(d + 1, s, m)\\
  &\leq 4 (d + 2) (s + m + 1) 2^{(d + 1) (s + m + 1) H(d, s, m)}\\
  &\leq 2^{d + 2} \cdot 2^{s + m + 1} \cdot 2^{(d + 1) (s + m + 1)
    H(d, s, m)}\\
  &\leq 2^{2 (d + 1) (s + m + 1)} \cdot 2^{(d + 1) (s + m + 1) H(d, s,
    m)}\\
  &\leq 2^{4 (d + 1) (s + m + 1) H(d, s, m)}\\
  & = 2^{H'(d, s, m)}\;.\qedhere
\end{align*}
\end{proof}

%% file: sec-lowb.tex
The rough pattern of our hardness proof 
resembles those by e.g.\ \ifomitproofs\else\citet{urquhart99,phs-mfcs2010}\fi,
where a fast-growing function is computed weakly,
then its result is used to allocate space for simulating a universal machine,
and finally the inverse of the function is computed weakly 
for checking purposes.
Indeed, we simulate Minsky machines whose counters are $\mathrm{tower}$-bounded,
but the novelty here is in the inverse computations.
Specifically, for each Minsky counter $c$, we maintain its dual $\hat{c}$ 
and simulate each zero test on $c$ by a split rule that launches a thread 
to check that $\hat{c}$ has the maximum value.
Recalling that such rules split all values non-deterministically,
we must construct the simulating system carefully so that 
such non-determinism cannot result in erroneous behaviours.

The auxiliary threads check that a counter is at least $\mathrm{tower}(k)$
by seeking to apply split rules at least $\mathrm{tower}(k - 1)$ times
along every branch.  The difficulty here is, similarly, 
how to count up to $\mathrm{tower}(k - 1)$ or more in a manner which is 
robust with respect to the non-determinism of the split rules.

A hierarchy of BVASS for the latter purpose is given in
\autoref{f:Bk}--- recall the depicting conventions in
\autoref{ssub-example}.  In this system, after the unary rule from
$q^{\mathrm{loop}}_k$ that decrements $d_k$, we have that $\?B_k$
behaves like $\?B_{k - 1}$ from state $q^{\mathrm{init}}_{k - 1}$.

\begin{figure}
  \ifomitproofs\begin{center}\else\centering\fi
  \begin{tikzpicture}[auto,node distance=1.5cm]
    \node[state,label=left:{$\?B_k$:}](qinitk){$q^{\mathrm{init}}_k$};
    \node[state,right=of qinitk](q1k){$q^1_k$};
    \node[state,right=of q1k,label=right:{$+$}](q2k){$q^2_k$};
    \node[state,right=of q2k](qloopk){$q^{\mathrm{loop}}_k$};
    \node[rectangle,dotted,draw=black!50,below left=1.1cm and -.9cm of
      qloopk,rounded corners=8pt,text width=1.2cm,text
      height=1.2cm]{$\?B_{k - 1}$};
    \node[state,below=1.2cm of qloopk](Bkm1){$q^{\mathrm{init}}_{k-1}$};
    \path[->,every node/.style={font=\footnotesize}]
      (qinitk) edge node{$\text{{+}{+}} d_{k - 1}$} (q1k)
          (q1k) edge[loop below] node{$\begin{array}{c}
                                      \text{{-}{-}} d_{k - 1} \\
                                      \text{{+}{+}} d'_{k - 1} \\
                                      \text{{+}{+}} d'_{k - 1}
                                      \end{array}$} ()
          (q1k) edge (q2k)
          (q2k) edge[loop below] node{$\begin{array}{c}
                                      \text{{-}{-}} d'_{k - 1} \\
                                      \text{{+}{+}} d_{k - 1}
                                      \end{array}$} ()
          (q2k) ++(.33,.33) edge[bend left] (qloopk)
          (q2k) ++(.33,-.33) edge[bend right] (qloopk)
          (qloopk) edge[bend angle=70,bend left] (qinitk)
          (qloopk) edge node{$\text{{-}{-}} d_k$} (Bkm1);
    \draw[semithick] (q2k) -- +(.33,.33)
                     (q2k) -- +(.33,-.33);
    \draw[color=black!40] (q2k) ++(.33,-.33) 
                          arc[start angle=-45,end angle=45,radius=.47];
    \node[state,below=2cm of qinitk,label=left:{$\?B_1$:}](qinit1){$q^{\mathrm{init}}_1$};
    \node[state,accepting by double,right=of qinit1](qleaf){$q^{\mathrm{leaf}}$};
    \path[->,every node/.style={font=\footnotesize}]
          (qinit1) edge node{$\begin{array}{c}
                              \text{{-}{-}} d_1;%
                              \text{{-}{-}} d_1
                              \end{array}$} (qleaf);
  \end{tikzpicture}
  \ifomitproofs\end{center}\fi
  \caption{\label{f:Bk} Defining $\?B_k$ for $k > 1$ (above),
                        and $\?B_1$ (below).}
\end{figure}

\begin{lemma}
\label{l:Bk}
For every $k \geq 1$ and vector of naturals $\vec v_0$
such that $\vec v_0(d_i) = \vec v_0(d'_i) = 0$ for all $i < k$,
we have that $\?B_k$ has a $(q^{\mathrm{init}}_k, \vec v_0)$-rooted
$\{q^{\mathrm{leaf}}\}$-leaf-covering deduction tree if and only if
$\vec v_0(d_k) \geq \mathrm{tower}(k)$.
\end{lemma}

\begin{proof}
We proceed by induction on $k$, 
where the base case $k = 1$ is immediate,
so let us consider $k > 1$ and $\vec v_0$
such that $\vec v_0(d_i) = \vec v_0(d'_i) = 0$ for all $i < k$.

If $\vec v_0(d_k) \geq \mathrm{tower}(k)$,
we observe that $\?B_k$ can proceed from $(q^{\mathrm{init}}_k, \vec v_0)$
as follows:
\begin{itemize}
\item
each loop at $q^1_k$ empties $d_{k - 1}$,
i.e.\ doubles $d_{k - 1}$ and transfers it to $d'_{k - 1}$;
\item
each loop at $q^2_k$ empties $d'_{k - 1}$,
i.e.\ transfers $d'_{k - 1}$ back to $d_{k - 1}$;
\item
each split from $q^2_k$ divides $d_{k - 1}$ into two equal values,
and divides $d_k$ into two values that differ by at most $1$.
\end{itemize}
In any deduction tree thus obtained, at every node which is 
the $h$th node with state label $q^{\mathrm{loop}}_k$ from the root, 
and whose vector label is $\vec w$, we have:
\begin{align}
\vec w(d_{k - 1}) & = h\:, &
\vec w(d'_{k - 1}) & = 0\:, &
\vec w(d_k) & \geq 2^{\mathrm{tower}(k - 1) - h}\:.
\end{align}
Hence, by returning control to $q^{\mathrm{init}}_k$ 
as long as the value of $d_k$ is at least $2$,
$\?B_k$ can reach along every vertical path 
a node with state label $q^{\mathrm{loop}}_k$ 
at which the values of $d_{k - 1}$ and $d_k$ 
are respectively equal to $\mathrm{tower}(k - 1)$ and at least $1$.  
To complete the deduction tree to be $\{q^{\mathrm{leaf}}\}$-leaf-covering, 
from every such node we let $\?B_k$ decrement $d_k$ 
and apply the induction hypothesis on $\?B_{k-1}$.

The interesting direction remains, 
so suppose $\?D$ is a $(q^{\mathrm{init}}_k, \vec v_0)$-rooted 
$\{q^{\mathrm{leaf}}\}$-leaf-covering deduction tree of $\?B_k$.
Since at every $q^{\mathrm{loop}}_k$-labelled node in $\?D$,
the value of $d_k$ must be at least $1$,
it suffices to establish the following claim
and apply it for the maximum $h$:
\begin{claim}
For each $0 < h \leq \mathrm{tower}(k - 1)$,
$\?D$ contains $2^h$ incomparable nodes (i.e., none is a descendant of another) 
whose state label is $q^{\mathrm{loop}}_k$ and 
at which $d_{k - 1} + d'_{k - 1}$ has value at most $h$.
\end{claim}
In turn, by induction on $h$, 
that claim is a straightforward consequence of the next one.
(For the base case of that induction, i.e.\ $h = 1$, 
apply the next claim with $h' = 0$.)
\begin{claim}
For each node $n$ in $\?D$ whose state label is $q^{\mathrm{init}}_k$ and 
at which $d_{k - 1} + d'_{k - 1}$ has some value $h' < \mathrm{tower}(k - 1)$,
there must be two incomparable descendants $n_1$ and $n_2$ 
whose state labels are $q^{\mathrm{loop}}_k$ and 
at which the values of $d_{k - 1} + d'_{k - 1}$ are at most $h' + 1$.
\end{claim}
Consider a node $n$ as in the latter claim.
After the increment of $d_{k - 1}$ and the loops at $q^1_k$ and $q^2_k$,
the value of $d_{k - 1} + d'_{k - 1}$ will be at most $2 (h' + 1)$.
If the first split divides $d_{k - 1} + d'_{k - 1}$ equally, we are
done.

Otherwise, we have a $q^{\mathrm{loop}}_k$-labelled descendant $n'$ of $n$ 
at which $d_{k - 1} + d'_{k - 1}$ has value at most $h'$.
In particular, $d_{k - 1}$ is less than $\mathrm{tower}(k - 1)$ at $n'$,
so recalling the induction hypothesis regarding $\?B_{k - 1}$, 
the child $n''$ of $n'$ cannot be $q^{\mathrm{init}}_{k - 1}$-labelled.
Thus, $n''$ must be $q^{\mathrm{init}}_k$-labelled,
and the value of $d_{k - 1} + d'_{k - 1}$ at $n''$ 
is the same as at $n'$, so at most $h'$.  
We can therefore repeat the argument with $n''$ instead of $n$, 
but since $\?D$ is finite, 
two incomparable descendants as required eventually exist.
\end{proof}

Relying on the properties of the BVASS $\?B_k$, we now establish the
hardness of lossy reachability, matching the membership in
\textsc{Tower} in \autoref{th:upb} already for BVASS.  Although we do
not match the upper bounds when the system dimension is fixed, we
remark that our simulation uses a number of counters which is linear
in the height of the tower of exponentials with coefficient~$2$.

\begin{theorem}
Reachability for lossy BVASS is \textsc{Tower}-hard.
\end{theorem}

\begin{proof}
For a notion of Minsky machines that is similar to 
how \abv\ were defined in \autoref{sub-abvass},
let such a machine be given by a finite set of states $Q$,
a finite set of counters $C$, and finite sets of
increment rules ``$q \xrightarrow{\text{{+}{+}}c} q_1$,''
decrement rules ``$q \xrightarrow{\text{{-}{-}}c} q_1$'' and
zero-test rules ``$q \xrightarrow{c \stackrel{?}{=} 0} q_1$.''
By simulating a tape using two stacks,
and simulating a stack using two counters, 
it is straightforward to verify that 
the following problem\ifomitproofs\else~\citep[called $\mathbf
F_3$-MM in][]{arXiv/Schmitz13}\fi\ is \textsc{Tower}-hard:
\begin{quote}
Given a Minsky machine $\?M$ and two states $q_0, q_H$,
does $\?M$ have a computation that 
starts in $q_0$ with all counters having value $0$, 
ends in $q_H$, and is such that all counter values are at most
$\mathrm{tower}(|\?M|)$?
\end{quote}

We establish the theorem by working with the equivalent
top-down coverability problem (see \autoref{s:lossy.reach}).
We show that, given a Minsky machine $\?M$ of size $K$
and two states $q_0, q_H$, then a BVASS $\?A(\?M)$, a state $q_r$ and a
finite set $Q_\ell\eqdef\{q_H,q^\mathrm{leaf}\}$
are computable in logarithmic space, such that 
$\?M$ has a $0$-initialised $\mathrm{tower}(K)$-bounded 
computation from $q_0$ to $q_H$ if and only if 
$\?A(\?M)$ has a $(q_r, \vec 0)$-rooted $Q_\ell$-leaf-covering 
deduction tree.

\begin{figure}
\centering
\begin{tikzpicture}[auto,node distance=1.5cm]
    \node[rectangle,dotted,draw=black!50,rounded corners=8pt,text width=8cm,text
      height=2.1cm] at (1.9,-.3) {$\mathrm{tower}(m)$};
    \node[rectangle,dotted,draw=black!50,fill=black!2,rounded corners=8pt,text width=4.75cm,text
      height=2cm] at (3.45,-.3) {$2^m$};
    \node[rectangle,dotted,draw=black!50,fill=black!4,rounded corners=8pt,text width=2.25cm,text
      height=1.8cm] at (4.6,-.3) {$2m$};
    \node[intermediate,accepting above](z0){};
    \node[intermediate,initial,left=of z0](z){};
    \node[intermediate,right=of z0](z1){};
    \node[intermediate,right=of z1](z2){};
    \node[intermediate,right=of z2](z3){};
    \path[->,every node/.style={font=\footnotesize}]
      (z) edge node {$\text{{+}{+}}c_0$} (z0)
      (z0) edge[bend left] node {$\text{{-}{-}}c_3;\text{{+}{+}}c_0$} (z1)
      (z1) edge[bend left] (z0)
      (z0) edge[loop below] node {$\text{{-}{-}}c_0;\text{{+}{+}}c_2$} ()
      (z1) edge[bend left] node {$\text{{-}{-}}c_2$} (z2)
      (z2) edge[bend left] (z1)
      (z1) edge[loop below] node {$\text{{-}{-}}c_0;\text{{+}{+}}c_1$} ()
      (z2) edge[bend left] node {$\text{{-}{-}}c_1$} (z3)
      (z3) edge[bend left] node {$\text{{+}{+}}c_0;\text{{+}{+}}c_0$} (z2);
  \end{tikzpicture}  
\caption{\label{fig-weaktower}A VASS weakly computing
$\mathrm{tower}(m)$, with input counter $c_3$ and output counter
$c_0$~\citep{MayrMeyer81}.  Counter~$c_3$ is initialised to $m$, and
the others to~$0$.}
\end{figure}
For each counter $c$ of $\?M$, there are three counters in $\?A(\?M)$
denoted $c, \hat{c}, c'$.  The initial part of $\?A(\?M)$ employs a
``weak Petri computer''~\citep{MayrMeyer81} for the $\mathrm{tower}$
function, namely a constant VASS with designated input and output
counters and initial and final states.  Given a natural number $m$ in
its input counter and starting in its initial state, it can compute
$\mathrm{tower}(m)$ in its output counter upon reaching its final
state, but non-deterministically may also compute a smaller value (but
never a larger one).  This is a standard construction, using weak
routines for $2m$ and $2^m$, which we depict
in \autoref{fig-weaktower}.\footnote{The reader puzzled by the
``$\text{{+}{+}}c_0$'' increments in \autoref{fig-weaktower} should
observe that $\mathrm{tower}(0)=2^0=1$ but $2\cdot 0=0$.}  By means of
the latter VASS, each counter $\hat{c}$ in $\?A(\?M)$ is initialised
to have value $\mathrm{tower}(K)$ or possibly smaller.  Recalling that
the auxiliary VASS is constant, a simple pattern for incorporating it
into $\?A(\?M)$ is to use fresh states and counters for each
$\hat{c}$.

The main part of $\?A(\?M)$ consists of simulating $\?M$ from $q_0$, 
using the translations of increments, decrements and zero tests 
in \autoref{f:Minsky}.
For the increments and decrements, 
$\?A(\?M)$ also performs the opposite operation on the hatted counter, 
thereby keeping the sums $c + \hat{c}$ constant.
For the zero tests, $\?A(\?M)$ attempts 
by two loops and using the primed counter, 
to copy the hatted counter to $d_K$ 
and then employ $\?B_K$ (see \autoref{f:Bk})
to verify that the latter is maximal (i.e., has value $\mathrm{tower}(K)$).
Thus, $\?A(\?M)$ also has counters 
$d_i$ for $0 < i \leq K$ and 
$d'_i$ for $0 < i < K$,
and more precisely a variant of $\?B_K$ is employed 
that has the same dimension as $\?A(\?M)$ 
(and does not use the extra counters).

\begin{figure}
  \ifomitproofs\begin{center}\else\centering\fi
  \begin{tikzpicture}[auto,node distance=1.5cm]
    \node[intermediate,label=left:{$c \stackrel{?}{=} 0$\/:}](z0){};
    \node[intermediate,right=of z0](z1){};
    \node[intermediate,right=of z1](z2){};
    \node[intermediate,right=of z2](z3){};
    \node[rectangle,dotted,draw=black!50,below right=.4cm and .6cm of
      z2,rounded corners=8pt,text width=1.2cm,text
      height=.70cm]{$\?B_{K}$};
    \node[state,below right=.6cm and 1.3cm of z2](zB){$q^{\mathrm{init}}_K$};
    \path[->,every node/.style={font=\footnotesize}]
          (z0) edge (z1)
          (z1) edge[loop below] node{$\begin{array}{c}
                                      \text{{-}{-}} \hat{c};%
                                      \text{{+}{+}} d_K;%
                                      \text{{+}{+}} c'
                                      \end{array}$} ()
          (z1) edge (z2)
          (z2) edge[loop below] node{$\begin{array}{c}
                                      \text{{-}{-}} c';%
                                      \text{{+}{+}} \hat{c}
                                      \end{array}$} ()
          (z2) edge (z3)
          (z2) edge[bend left=10] (zB);
    \draw[color=black!40] (z2) ++(.6,0) 
                          arc[start angle=0,end angle=-23,radius=.6];
    \draw (z2) ++(.75,-.13) node{$+$};
    \node[intermediate,above=.6cm of z0,label=left:{$\text{{+}{+}} c$\/:}](i0){};
    \node[intermediate,right=of i0](i1){};
    \path[->,every node/.style={font=\footnotesize}]
          (i0) edge[swap] node{$\begin{array}{c}
                                \text{{+}{+}} c;%
                                \text{{-}{-}} \hat{c}
                                \end{array}$} (i1);
    \node[intermediate,right=of i1,label=left:{$\text{{-}{-}} c$\/:}](d0){};
    \node[intermediate,right=of d0](d1){};
    \path[->,every node/.style={font=\footnotesize}]
          (d0) edge[swap] node{$\begin{array}{c}
                                \text{{-}{-}} c;%
                                \text{{+}{+}} \hat{c}
                                \end{array}$} (d1);
  \end{tikzpicture}\vspace*{-.8em}
  \ifomitproofs\end{center}\fi
  \caption{\label{f:Minsky} Simulating the Minsky operations.}
\end{figure}

For each $0$-initialised $\mathrm{tower}(K)$-bounded 
computation of $\?M$ from $q_0$ to $q_H$, it is straightforward to check that
$\?A(\?M)$ can simulate it%
\ as follows:
\begin{itemize}
\item
each counter $\hat{c}$ is initialised to $\mathrm{tower}(K)$;
\item
in every simulation of a zero test $c \stackrel{?}{=} 0$,
the values of $c, \hat{c}, c', d_K$ 
are resp.\ $0, \mathrm{tower}(K), 0, 0$
before the two loops,
and $0, \mathrm{tower}(K), 0, \mathrm{tower}(K)$
before the split;
\item
at every start of $\?B_K$, 
the value of $d_K$ is $\mathrm{tower}(K)$
and all other counters have value $0$.
\end{itemize}
By \autoref{l:Bk}, we obtain a $(q_r, \vec 0)$-rooted 
$Q_\ell$-leaf-covering deduction tree of $\?A(\?M)$.

The other direction is more involved:
  we show that, if $\?A(\?M)$ has a $(q_r, \vec 0)$-rooted $Q_\ell$-leaf-covering
  deduction tree $\?D$, then $\?M$ has a $0$-initialised
  $\mathrm{tower}(K)$-bounded computation from $q_0$ to $q_H$.
\input{proof-cl-mm-abvass}
\ifomitproofs\else We conclude that $\?A(\?M)$ has the required properties.\fi
\end{proof}

Since lossy reachability reduces to reachability and by
\autoref{prop-abvass-ll} and \autoref{prop-abvass-llw}, this entails:
\begin{corollary}\label{cor-mell}
  Provability in MELL, MELLW, and LLW is \textsc{Tower}-hard.
\end{corollary}

%% file: proof-cl-mm-abvass.tex
By construction, $\?D$ consists of a path $\pi$ 
from which there are branchings to deduction trees of $\?B_K$.
The main part of $\pi$ consists of the simulations of
increments, decrements and zero tests as in \autoref{f:Minsky}.
From it, we obtain a $0$-initialised $\mathrm{tower}(K)$-bounded 
computation of $\?M$ from $q_0$ to $q_H$,
after observing the following for every counter $c$ of $\?M$:
\begin{itemize}
\item
After $\hat{c}$ is initialised in $\?D$,
the value of $c + \hat{c} + c'$ is always at most $\mathrm{tower}(K)$.
\item
For each simulation of a zero test of $c$, we have by \autoref{l:Bk} 
that the value of $d_K$ is $\mathrm{tower}(K)$ before the split
and is $0$ after the split on the path $\pi$,
and consequently that the values of $c, \hat{c}, c'$ 
are $0, \mathrm{tower}(K), 0$ (respectively) before the two loops.
\item
The value of $c$ may erroneously decrease due to the branchings,
but since that makes the value of $c + \hat{c} + c'$ 
smaller than $\mathrm{tower}(K)$,
such losses may occur only after the last simulation of a zero test of $c$,
and so cannot result in an erroneous such simulation.
\item
Similarly, only the last transfer of $c'$ to $\hat{c}$ may be incomplete 
(i.e., it does not empty~$c'$).
\end{itemize}

%% file: sec-incr.tex
We investigate in this section the complexity of reachability in
increasing or expansive \abv.  The latter is related to provability in
contractive linear logic, and as shown by \citet{urquhart99}, to
provability in the conjunctive-implicative fragment of relevance
logic, and we treat it in \autoref{sub-exp}.

\subsection{Increasing Reachability}\label{sub-incr}
Expansive \abv\ are not quite dual to lossy \abv: the natural model
for this is that of \emph{increasing} \abv, which feature additional
deduction rules
\begin{gather*}
  \drule{q,\vec v}{q,\vec v+\vec e_i}{increase}
\end{gather*}
for all $q$ in $Q$ and $0<i\leq d$.

\subsubsection{Bottom-Up Coverability\nopunct.} As with lossy reachability,
increasing reachability corresponds to a coverability problem in a variant
of \abv.  Let us define a variant of \abv\ that
provides different semantics
\begin{itemize}
\item for rules in $T_f$ as \emph{meets} ``$q\to q_1\sqcap q_2$''
  instead of forks, and
\item for rules in $T_z$ as \emph{zero-jumps}
  ``$q\xrightarrow{\vec\omega}q_1$'' instead of full zero tests.
\end{itemize}
Their semantics are now defined by the deduction rules
\begin{equation*}
  \drule{q,\vec v_1\sqcap\vec v_2}{q_1,\vec v_1\quad q_2,\vec v_2}{meet}
  \qquad\drule{q,\vec 0}{q_1,\vec v}{zero-jump}
\end{equation*}
where the \emph{meet} $\vec v_1\sqcap\vec v_2$ of two vectors in
$\+N^d$ is the component-wise minimum of $\vec v_1$ and $\vec
v_2$: for all $0<i\leq d$,
\begin{equation}
  (\vec v_1\sqcap\vec v_2)(i)\eqdef\min(\vec v_1(i),\vec v_2(i))\;.
\end{equation}%

Let us call the resulting model ABVASSi.  Given an ABVASSi $\?A$, a
state $q_r$, and a finite set of states $Q_\ell$, the \emph{bottom-up
  coverability} or \emph{root coverability} problem asks for the
existence of a deduction tree $\?D$ with root label $(q_r,\vec v)$ for
some $\vec v$ in $\+N^d$ where $d$ is the dimension of $\?A$, and with
every leaf labelled by some element of $Q_\ell\times\{\vec 0\}$.

By a reasoning similar to the one employed for top-down coverability
in ABVASSr, bottom-up coverability in ABVASSi corresponds to
increasing reachability in \abv: by monotonicity we can always
increase as soon as possible in the latter, either at the root, or
right after a full zero test, or right after an ``imbalanced''
fork---where increases differ on the two branches.

\subsubsection{Pseudo-Increasing \abv\nopunct.}\label{sub-pincr}
Let us consider yet another variant of increasing \abv, which will be
used in the complexity analysis, and which combines increasing steps
with unary steps.  Given a vector $\vec u$ in $\+Z^d$, let us denote
by $\vec u_-$ the vector in $\+N^d$ defined for all
$0<i\leq d$ by
\begin{align}
  \vec u_-(i)&\eqdef\begin{cases}-\vec u(i)&\text{if }\vec u(i)<0\\
  0&\text{otherwise}.\end{cases}
\intertext{Then, for a vector $\vec v$ in $\+N^d$, the
\emph{join} $\vec v\sqcup\vec u_-$ defined for all $0<i\leq d$ by}
  (\vec v\sqcup\vec u_-)(i)&\eqdef\max(\vec v(i),\vec u_-(i))
\end{align}
is the minimal vector greater or equal to $\vec v$ that allows to fire
a unary rule with vector~$\vec u$.

A \emph{pseudo-increasing} \abv\ does not have the increasing rule,
but uses instead a different semantics for its unary rules
$q\xrightarrow{\vec u}q'$ in $T_u$, which can be used for \emph{any}
$\vec v$ in $\+N^d$:
\begin{equation*}
  \drule{q,\vec v}{q',(\vec v\sqcup\vec u_-)+\vec u}{pseudo-unary}
\end{equation*}
The idea of the pseudo-unary rule is that it implicitly applies the
minimal amount of increase necessary to use a given unary rule.

Reachability (from $(q_r,\vec 0)$ to $Q_\ell\times\{\vec 0\}$) in an
increasing \abv\ is then equivalent to reachability in the same \abv\
with pseudo-increasing semantics.  We can indeed delay increases
occurring right before another rule: an increase right before a fork
or a split can be performed after it, no increase can occur right
before a full zero test, and \emph{superfluous} increases right before
a unary rule $q\xrightarrow{\vec u}q'$ can be performed after it,
i.e.\ an increase from $(q,\vec v)$ to $(q,\vec v+\vec e_i)$ with
$\vec v\geq\vec u_-$ can rather be performed from $(q',\vec v+\vec u)$
to $(q',\vec v+\vec e_i+\vec u)$.  The remaining increases right
before unary rules become part of pseudo-unary rules.

\subsection{Complexity of Increasing Reachability}
In the more restricted case of BVASS, which do not feature forks nor
full zero tests, bottom-up coverability coincides with increasing
reachability, and this problem was called more simply ``coverability''
by \citet{verma05}:
\begin{fact}[\citet{demri12}]
  Reachability in increasing BVASS is \textsc{2-ExpTime}-complete.
\end{fact}

Since increasing \abv\ are not too different from expanding \abv, the
fact that their complexity is the same is not too surprising:
\begin{theorem}\label{fc-avass-rcov}
  Reachability in increasing AVASS and increasing \abv\ is
  \Ack-complete.
\end{theorem}
\begin{proof}
  For the lower bound, a reduction from the reachability problem in
  \emph{increasing} Minsky machines~\citep{phs-mfcs2010} is
  straightforward, since it uses the same encoding of zero tests as
  the proof sketch for \autoref{fc-avass-reach}.

  For the upper bound, by \autoref{lem-abvass}, we can restrict
  ourselves to ABVASS without loss of generality.  Define the partial
  order $\leq$ over configurations in $Q\times\+N^d$ by $(q,\vec
  v)\leq(q',\vec v')$ if $q=q'$ and $\vec v\leq\vec v'$; this is the
  product ordering over $Q\times\+N^d$, where $Q$ is ordered using
  equality.  By Dickson's Lemma, $(Q\times\+N^d,\leq)$ is a \emph{well
    quasi order}.  The proof for the upper bound thus follows the
  typical steps of an application of the \emph{length function
    theorem} for Dickson's Lemma found in \citep{FFSS-lics2011}.

  Let us consider an increasing reachability witness $\?D$ and the
  sequence of labels $(q_r,\vec 0)=(q_0,\vec v_0),(q_1,\vec
  v_1),\dots,(q_m,\vec v_m)\in Q_\ell\times\{\vec 0\}$ along some
  branch $n_0,n_1,\dots,n_m$ of $\?D$.  By ignoring the intermediate
  increase steps, we extract a \emph{pseudo-subsequence} $(q_{i_0},\vec
  v_{i_0}),(q_{i_1},\vec v_{i_1}),\dots,(q_{i_p},\vec v_{i_p})$ with
  $(q_{i_0},\vec v_{i_0})\geq (q_r,\vec 0)$ and, for each $0<j\leq p$,
  $(q_{i_j},\vec v_{i_j})\geq (q_{i_{j-1}+1},\vec v_{i_{j-1}+1})$.

  Assume now that $\?D$, among all the increasing reachability
  witnesses, has pseudo-subsequences of minimal length (noted $p+1$
  above) along each branch.  Then, along any branch, for all
  $0\leq j<k\leq p$, $(q_{i_j},\vec v_{i_j})\not\leq(q_{i_k},\vec
  v_{i_k})$, or a sequence of increases would allow to go from
  $n_{i_{j-1}}$ to $n_{i_k}$ directly with a strictly shorter
  pseudo-subsequence.  In terms of the wqo, this means that, along any
  branch, the pseudo-subsequence is a \emph{bad} sequence.  Let us
  furthermore apply the strategy in \autoref{sub-pincr} and delay
  increases as much as possible---note that this might provide further
  opportunities for reducing the length of pseudo-sequences along the
  branches of $\?D$.  The resulting increasing reachability witness,
  where the remaining increases occur necessarily just before unary
  rules, can then be seen as a \emph{pseudo-increasing} reachability
  witness, where every sequence of labels along every branch is a bad
  sequence for the wqo $(Q\times\+N^d,\leq)$.

  Define the \emph{norm} $\|q,\vec v\|$ of a configuration $(q,\vec
  v)$ in $Q\times\+N^d$ as the infinite norm $\max_{0<i\leq d}\vec
  v(i)$ of $\vec v$.  Observe that, along any branch of a
  pseudo-increasing reachability witness, if $(q_j,\vec v_j)$ and
  $(q_{j+1},\vec v_{j+1})$ are two successive labels, then
  \begin{equation}
    \|q_{j+1},\vec v_{j+1}\|\leq\|q_j,\vec
    v_j\|+\mathrm{max}^-(T_u)+\mathrm{max}^+(T_u)\;.
  \end{equation}
  Define accordingly $g(x)\eqdef
  x+\mathrm{max}^-(T_u)+\mathrm{max}^+(T_u)$, then for the $j$th label
  along a branch, $\|q_j,\vec v_j\|\leq g^j(0)$ the $j$th iterate of
  $g$.  This shows that the sequence of labels along every branch of
  our pseudo-increasing reachability witness is a bad sequence
  \emph{controlled} by $(g,0)$.

  A \emph{length function theorem} for a wqo is a combinatorial
  statement bounding the length of bad controlled sequences.  In our
  case, for the wqo \mbox{$(Q\times\+N^d,\leq)$} and the control $(g,0)$, the
  theorem in~\citep{FFSS-lics2011} yields an
  \begin{equation}\label{eq-lft}
    F_{d+1}\big(p(\mathrm{max}^-(T_u)+\mathrm{max}^+(T_u),
    |Q|)\big)\leq\mathrm{Ack}(p'(|\?A|))
  \end{equation}
  upper bound on the length of branches for some polynomial functions
  $p$ and $p'$, where $(F_d{:}\,\+N\to\+N)_d$ is a hierarchy of
  fast-growing functions~\citep{lob70} with $\mathrm{Ack}(n)\eqdef
  F_{n+1}(n)$.  A non-deterministic combinatorial algorithm can thus
  compute the bound in \eqref{eq-lft} and attempt to find a
  pseudo-increasing witness of such bounded height (note that the
  branching degree of witnesses is also bounded) in Ackermannian time.
  As with lossy reachability, the main parameter in this complexity
  upper bound is the dimension $d$ of the \abv.
\end{proof}

\subsection{Complexity of Expansive Reachability}\label{sub-exp}

%% file: sec-ack.tex
Turning to expansive reachability, we present now the missing proof of
\autoref{th-abvass-cov}:
\abvasscov*
\begin{proof}
  The lower bound is proved by \citet{urquhart99}.  The upper bound is
  also similar to that of \citeauthor{urquhart99} for provability in
  LR+, and follows essentially the same scheme as in the increasing
  case in the proof of \autoref{fc-avass-rcov}.  By
  \autoref{lem-abvass} we restrict ourselves to ABVASS.  Define the
  partial order $\sqsubseteq$ over configurations in $Q\times\+N^d$ by
  $(q,\vec v)\sqsubseteq (q',\vec v')$ if $q=q'$, $\vec v\leq\vec v'$,
  and $\sigma(\vec v)=\sigma(\vec v')$, where $\sigma(\vec
  v)\eqdef\{0<i\leq d\mid\vec v(i)>0\}$ denotes the \emph{support} of
  $\vec v$.  The quasi-order $(Q\times\+N^d,\sqsubseteq)$ is
  isomorphic to the sub-order of the product ordering over
  $Q\times\+N^d\times 2^d$ induced by the restriction to triples
  $(q,\vec v,s)$ where $\sigma(\vec v)=s$, and is therefore a wqo by
  Dickson's Lemma.

  Substituting $\sqsubseteq$ for $\leq$, we show as in the proof of
  \autoref{fc-avass-rcov} that, if there is an expansive reachability
  witness $\?D$, then there is one where pseudo sequences are bad
  sequences for $(Q\times\+N^d,\sqsubseteq)$ along each branch, and
  where expansions are applied as late as possible.  Thus $\?D$ can be
  seen as a \emph{pseudo-expansive} reachability witness, using the
  semantics of pseudo-unary rules for unary rules in $T_u$, with the
  additional restriction that such a rule can only be applied if
  $\sigma(\vec u_-)\subseteq\sigma(\vec v)$.  This restriction
  reflects the fact that expansions cannot increase a zero coordinate
  in $\vec v$.  The remaining steps are the same as in the proof of
  \autoref{fc-avass-rcov}: the sequences of labels along the branches
  of $\?D$ are bad $(g,0)$-controlled sequences for
  $(Q\times\+N^d,\sqsubseteq)$, and we obtain similarly an
  \begin{equation}\label{eq-lft-exp}
    F_{d+1}\big(p(\mathrm{max}^-(T_u)+\mathrm{max}^+(T_u),|Q|\cdot
    2^d)\big)\leq\mathrm{Ack}(p'(|\?A|))
  \end{equation}
  upper bound on the height of our witness, for some polynomial
  functions $p$ and $p'$.  %
  Again,
  the main complexity parameter is the dimension $d$ of the \abv.
\end{proof}

\begin{corollary}
  MALLC and LLC provability are \textsc{Ackermann}-complete.
\end{corollary}
\begin{proof}
  By \autoref{th-abvass-cov} and the reductions from LLC provability
  to \abv\ expansive reachability in \autoref{prop-ll-abvass} and from
  AVASS expansive reachability to MALLC provability
  in \autoref{prop-abvass-llc}.
\end{proof}

%% file: sec-concl.tex
\begin{table}
  \tbl{\label{tab-ll}The
  complexity of provability in fragments and variants of LL.}{
  \begin{tabular}{lccc}
  \toprule
                & MELL                                   & LL           \\
  \midrule
                & \textsc{Tower}-hard, $\Sigma_1^0$-easy & $\Sigma_1^0$-c.~\citep{lincoln92} \\
  with W        %
                &\textsc{Tower}-c.                      &\textsc{Tower}-c.\\
  with C        %
                &\textsc{2Exp}-c.~\citep{schmitz14}     & \textsc{Ack}-c. \\
  \bottomrule    
  \end{tabular}}
\end{table}
Although connections between propositional linear logic and families
of counter machines have long been known, they have rarely been
exploited for complexity-theoretic results%
.  Using a model of alternating branching VASS, we have unified
several of these connections, and derived complexity bounds for
provability in substructural logics from the (old and new) bounds on
\abv\ reachability, summarised in \autoref{tab-ll} and
\autoref{tab-concl} respectively.

\begin{table}
  \tbl{\label{tab-concl}The complexity of reachability problems in
  \abv.}{
  \begin{tabular}{lccc}
  \toprule
                & AVASS                                & BVASS                                 & \abv           \\
  \midrule
  Reachability  & $\Sigma_1^0$-c.~\citep{lincoln92}    & \textsc{Tower}-hard, $\Sigma_1^0$-easy & $\Sigma_1^0$-c. \\
  Lossy reach.  & \textsc{2Exp}-c.~\citep{courtois14} &\textsc{Tower}-c.                       &\textsc{Tower}-c.\\
  Incr.\ reach. & \textsc{Ack}-c.~\citep{urquhart99}  &\textsc{2Exp}-c.~\citep{demri12}        & \textsc{Ack}-c. \\
  \bottomrule    
  \end{tabular}}
\end{table}

%% file: sec-mell.tex
Our main results in this regard are the \textsc{Tower}-completeness of
provability in LLW and the new \textsc{Tower} lower bound for MELL:
the latter has consequences on numerous problems mentioned
in \autoref{sec-abvass}, and entails for instance that the
satisfiability problem for FO$^2$ on data trees is
non-elementary~\citep{bojanczyk09,dimino13}.
The \textsc{Ackermann}-completeness of MALLC and LLC is perhaps less
surprising in the light of \citeauthor{urquhart99}'s results, but we
take it as a testimony to the versatility of the \abv\ model.

The main open question remains whether BVASS reachability, or
equivalently MELL provability, is decidable.